\RequirePackage{fix-cm}
\documentclass[smallextended]{svjour3}

\smartqed

\usepackage{graphicx,hyperref,amscd,amssymb,amsmath,latexsym,enumitem,blkarray,euscript}

\journalname{Mathematical Physics, Analysis and Geometry}

\newcommand{\C}{\mathbb C}
\newcommand{\R}{\mathbb R}
\newcommand{\N}{\mathbb N}

\newcommand{\Tr}{\operatorname{Tr}}

\newcommand{\pr}{\operatorname{pr}}
\newcommand{\Span}{\operatorname{Span}}
\newcommand{\Range}{\operatorname{Range}}

\newcommand{\la}{\langle}
\newcommand{\ra}{\rangle}

\newcommand{\hi}{\mathcal H}
\newcommand{\bi}{\mathcal B}
\newcommand{\ki}{\mathcal K}
\newcommand{\ai}{\mathcal A}
\newcommand{\si}{\mathcal S}
\newcommand{\ti}{\mathcal T}
\newcommand{\li}{\mathcal L}
\newcommand{\Ni}{\mathcal N}
\newcommand{\Ri}{\mathcal R}

\begin{document}

\title{The induced semigroup of Schwarz maps to the space of Hilbert-Schmidt operators}

\titlerunning{The induced semigroup of Schwarz maps}

\author{George Androulakis \and Alexander Wiedemann \and Matthew Ziemke}

\authorrunning{Androulakis, Wiedemann, Ziemke}

\institute{George Androulakis \at
              Department of Mathematics, University of South Carolina, Columbia, SC \\
              \email{giorgis@math.sc.edu}           
           \and
           Alexander Wiedemann \at
           Department of Mathematics, University of South Carolina, Columbia, SC \\
           \email{alwiedemann@davidson.edu}  \\ \emph{Present address: Department of Mathematics and Computer Science, Davidson College, Davidson, NC}
           \and
           Matthew Ziemke \at
           Department of Mathematics,	Drexel University, Philadelphia, PA \\
           \email{mjz55@drexel.edu}
}

\date{Received: date / Accepted: date}

\maketitle

\begin{abstract} We prove that for every semigroup of Schwarz maps on the von~Neumann algebra of all bounded linear operators on a Hilbert space which has a subinvariant faithful normal state there exists an associated semigroup of contractions on the space of Hilbert-Schmidt operators of the Hilbert space. Moreover, we show that if the original semigroup is weak$^*$ continuous then the associated semigroup is strongly continuous. We introduce the notion of the extended generator of a semigroup on the bounded operators of a Hilbert space with respect to an orthonormal basis of the Hilbert space. We describe the form of the generator of a quantum Markov semigroup on the von~Neumann algebra of all bounded linear operators on a Hilbert space which has an invariant faithful normal state under the assumption that the generator of the associated semigroup has compact resolvent.
\keywords{Schwarz maps \and strongly continuous semigroups \and generator of a semigroup \and form generator \and faithful normal states \and Moore-Penrose inverse}
\subclass{47D03 \and 47D07 \and 46L57 \and 81Q80}
\end{abstract}

\section{Introduction}

It is known that, under certain assumptions, semigroups on 
von~Neumann algebras or their preduals give rise to associated semigroups on Hilbert spaces. Moreover, 
these associated semigroups often have
stronger continuity properties than the original semigroups. For example, in 
\cite[Equation (2.1)]{kfgv} it is stated that if 
$(T_t)_{t \geq 0}$ is a quantum Markov semigroup on a 
von~Neumann algebra $\ai$ which has an invariant 
faithful normal state, and if $( \ki , \pi , \Omega )$
is the GNS triple associated to that state, then there exists a strongly continuous semigroup $(\overline{T_t})_{t \geq 0}$
of contractions on $\ki$ such that		
\begin{equation} \label{E:kfgv}
\overline{T_t} (\pi (A) \Omega) = \pi(T_t(A)) \Omega \quad 
\textrm{for all }A \in \ai \textrm{ and }t \geq 0.
\end{equation}
Since the proof of this statement is not included in
\cite{kfgv} we provide a proof here (see Remarks~\ref{R:kfgv} and \ref{R:morekfgv}). Other results which give rise to semigroups on Hilbert spaces starting from semigroups defined on spaces of operators
can be found in literature. For example, in \cite[Footnote of Theorem 6]{k} it is proved that every strongly continuous 
semigroup $(T_t)_{t \geq 0}$ of positive isometries on the real Banach space of self-adjoint trace-class operators on a Hilbert space 
gives rise to a strongly continuous semigroup $(V_t)_{t \geq 0}$ of isometries on the Hilbert space such that $T_t$ is given as a 
conjugation by $V_t$ for all $t \geq 0$. In \cite[Theorem 3]{el1} it is proved that, under appropriate assumptions, weakly continuous
semigroups on $\bi (\hi )$ (where $\hi$ is a separable Hilbert space) give rise to corresponding semigroups of unitaries 
on some associated Hilbert space. In
\cite[Theorem~3.3.7]{holevo2001} the author produces a strongly continuous group of 
unitaries associated with a norm continuous semigroup on the space of 
trace-class operators on a related Hilbert space.

In this work we prove a result similar to the result stated above in Equation~(\ref{E:kfgv}) (Theorem~\ref{theorem4}).
More precisely, we prove that every semigroup of Schwarz maps 
on $\bi (\hi )$ (where $\hi$ is a Hilbert space) which has
an invariant faithful state gives rise to an
associated semigroup $(\widetilde{T_t})_{t \geq 0}$ 
of contractions on the space of 
Hilbert-Schmidt operators on $\hi$. Our map is ``more symmetric"
than the one provided by Equation~(\ref{E:kfgv})
(see the comments following Remark~\ref{R:kfgv}). We introduce the 
notion of the extended generator of a semigroup on bounded operators
on a Hilbert space with respect to an orthonormal basis of the
Hilbert space, and we explicitly describe how the generators of $(T_t)_{t \geq 0}$
and $(\widetilde{T_t})_{t \geq 0}$ and the extended generator of $(T_t)_{t\geq 0}$ are related. We apply these descriptions to a quantum Markov semigroup $(T_t)_{t\geq0}$ having an invariant faithful normal 
state under the assumption that the generator of $(\widetilde{T_t})_{t \geq 0}$ has compact resolvent, which allows us to describe	the form of the extended generator (and thereby the generator) of the semigroup $(T_t)_{t \geq 0}$
with respect to an orthonormal basis, and thereby the generator itself (see \ref{Thm:final2}).

\subsection{Structure}
\noindent$\bullet$ In Section~\ref{prelims} we establish formal notation and definitions, and give some historical notes on the terminology.

\noindent$\bullet$ In Section~\ref{construct} we consider several constructions arising from faithful, positive, normal functionals. In particular, in Subsection~\ref{inducing} we prove that every faithful positive normal functional on $\bi (\hi)$ induces a canonical bounded linear map from $\bi (\hi )$ to $\si_2 (\hi )$. This map is used in Theorem~\ref{proposition3} to prove that for every bounded linear Schwarz map on $\bi (\hi )$, which has a subinvariant faithful positive functional, there exists a corresponding contraction on $\si_2 (\hi )$. In Subsection~\ref{alternate} we consider an alternate construction for such induced maps using the GNS construction, and then compare and contrast the two methods.

\noindent$\bullet$ In Section~\ref{schwarz} we recall the basic notions of continuity for semigroups, as well as formalize the definition of a semigroup's generator and its generator's domain. In Subsection~\ref{extending} we introduce the notion of an extended generator, which can be defined on a larger domain while still agreeing with the usual generator on all finite subspaces. Theorem~\ref{theorem4} relates the domains and actions of the generator, the extended generator, and the generator of the semigroup induced on $\si_2 (\hi )$.

\noindent$\bullet$ In Section~\ref{QMS} we investigate the applications of Theorem~\ref{theorem4} in the study of quantum Markov semigroups (QMSs), for which the exact form of the generator is known if the generator is bounded (see \cite{GKS} and \cite{L}). In Subsection~\ref{compactresolvent}, we describe the form of a QMS generator in the case that the generator of the semigroup induced on $\si_2 (\hi )$ has compact resolvent.

\section{Preliminaries}\label{prelims}

We first fix some notation.
If $\hi$ is a Hilbert space, let $(\mathcal{B}(\hi), \| \cdot \|_{\infty})$ denote the space of all bounded linear 
operators on $\hi$.  For $1 \leq p < \infty$, let $(\mathcal{S}_p(\hi), \| \cdot \|_p)$ denote the Schatten-$p$ space of 
operators.  In particular, $(\mathcal{S}_2(\hi), \| \cdot \|_2)$ denotes the space of Hilbert-Schmidt operators on $\hi$  and $(\mathcal{S}_1(\hi), \| \cdot \|_1)$ 
denotes the space of trace-class operators on $\hi$. Let 
$\la \cdot , \cdot \ra_{\si_2(\hi )}$ denote the inner product in $\mathcal{S}_2(\hi)$. If $L$ is a linear operator which is not
necessarily bounded, then $D(L)$ will denote the domain of $L$.

We adopt the convention that 
\textbf{functional} will always mean bounded linear functional. Usually the functionals that
we will consider will be faithful, positive, and normal, so this convention will help
us to cut down the number of adjectives.

We would like to recall the Schwarz inequality and define the Schwarz maps. 
The classical Cauchy-Schwarz 
inequality states that $| \la y, x \ra | \leq \| y \| \| x \|$ for all vectors $x,y$ in a Hilbert space.
This inequality is extended to $| \phi (y^*x) | \leq \sqrt{\phi (y^*y)}\sqrt{\phi (x^*x)}$
for all $x,y$ in a $C^*$-algebra $\ai$, where $\phi$ is a positive functional on 
$\ai$ (see \cite[Theorem~4.3.1]{kr}).
The last inequality can be further extended to $(T(y^*x))^* T(y^*x) \leq \| T(y^*y) \| T(x^*x)$
if $T$ is a completely positive map from a $C^*$-algebra $\ai$ to the $C^*$-algebra 
$\bi (\hi)$ of all bounded operators on a Hilbert space $\hi$ (see \cite[Lemma~2.6]{b}).
If in the last inequality one assumes that $\ai$ is unital and $T$ is unital, then by replacing 
$y$ by the unit we obtain 
\begin{equation} \label{E:SI}
T(x)^* T(x) \leq T(x^*x) \quad \textrm{for all }x \in \ai.
\end{equation}
A similar 
inequality was proved by Choi \cite[Corollary~2.8]{c} who proved that if $\ai$ is a 
unital $C^*$-algebra and $T$ is a $2$-positive unital map from $\ai$ to $\ai$ then 
$T(x^*)T(x) \leq T(x^*x)$ for all $x \in \ai$. Choi calls the last inequality 
``Schwarz inequality''. Similar inequalities appear in \cite[Theorem~1]{ka} and
\cite[Theorem~7.4]{s}. Since a positive linear map $T$ on a 
$C^*$-algebra $\ai$ satisfies $T(x^*)=T(x)^*$ for all $x \in \ai$, the last inequality is equivalent to (\ref{E:SI}).
Following \cite[page 14]{s2}, we say that a bounded linear operator $T$ on a 
$C^*$-algebra 
$\ai$ is a \textbf{Schwarz map} if Inequality~(\ref{E:SI}) is satisfied. The advantage of
Inequality~(\ref{E:SI}), versus the inequality proved by Choi, is that Inequality~(\ref{E:SI}) 
implies that $T$ is positive. Be warned that Inequality~(\ref{E:SI})
is not homogeneous for $T$, and therefore by scaling the operator $T$ by a positive
constant the above inequality is affected.

Next we recall the definition of invariant  functionals	and we define the notion of subinvariant positive functionals on a $C^*$-algebra. If $X$ is a Banach  space, $T :X \to X$ is a bounded linear operator, and $\omega$ is a functional on $X$, then $\omega$ is called \textbf{invariant for }~$\mathbf{T}$ if
\[
\omega (Tx)= \omega (x) \quad \textrm{for all }x \in X.
\]
If $\ai$ is a $C^*$-algebra, $T: \ai \to \ai$ is a positive bounded linear operator, and $\omega$ is a positive functional on $\ai$, then we will
say that $\omega$ is \textbf{subinvariant} for $T$ if 
\[ \omega (Ta) \leq \omega (a) \quad \textrm{for all }a \in \ai \textrm{ with } a \geq 0.\]
If $\hi$ is a Hilbert space, then a functional $\omega$ on 
$\bi (\hi )$ is called \textbf{normal} 
if and only if it is positive and continuous in the w$^*$ topology. This is 
equivalent to the fact that there exists a unique positive operator $\rho \in \si_1(\hi)$ such that 
\begin{equation} \label{E:normalstate}
\omega (x) = \Tr (\rho x) \quad \textrm{for all }x \in \bi (\hi )
\end{equation}
where $\Tr$ denotes the trace. The positive functional $\omega$ associated to 
the positive trace-class operator $\rho$ via Equation~(\ref{E:normalstate}) is denoted
by $\omega_\rho$. If $\omega$ is a state (i.e. unital positive functional) on $\bi (\hi )$
then $\omega$ is
normal if and only if the positive trace-class operator $\rho$ which satisfies
Equation~(\ref{E:normalstate}) 
has trace equal to $1$. Note that if $\hi$ is a Hilbert space
and $T:\bi (\hi) \to \bi (\hi )$ is a bounded linear operator, then a normal positive functional 
$\omega_\rho$ (for some positive trace-class operator $\rho$) is invariant for $T$ if and 
only if 
\[
T^\dagger (\rho) = \rho ,
\]
where $T^\dagger $ denotes the Banach dual operator of $T$ restricted to $\si_1(\hi )$
(viewed as a subspace of the dual of $\bi (\hi )$). Also, if $\hi$ is a Hilbert space and
$T:\bi (\hi) \to \bi (\hi)$ is a positive bounded linear operator, then a normal positive 
functional $\omega_\rho$ (for some positive trace-class operator $\rho$)
is subinvariant for $T$ if and only if 
\[
T^\dagger (\rho) \leq \rho .
\]

If $\hi$ is a Hilbert space, recall that a positive functional $\omega$ on $\bi (\hi)$ is
\textbf{faithful} provided 
that $\omega (x) >0$ for all $x >0$.  It is worth noting that $\bi (\hi)$ has a faithful
normal functional if and only if $\hi$ is separable (see \cite[Example 2.5.5]{br}). 

\section{Constructions Using Faithful, Positive, Normal Functionals}\label{construct}

We extensively use the next proposition, so we want to give it along with a proof.

\begin{proposition}\label{proposition1}
	Let $\hi$ be a Hilbert space and $\rho \in \mathcal{S}_1(\hi)$ be positive.  Then the following are equivalent:
	\begin{enumerate}
		\item[(i)] the positive normal functional  $\omega_{\rho}$ is faithful,
		
		\item[(ii)] the operator $\rho$ is injective, 
		
		\item[(iii)] the operator $\rho$ has dense range.
		
	\end{enumerate}
\end{proposition}

\begin{proof}
	$[(i)\Rightarrow (ii)]$.  Suppose $\omega_{\rho}$ is faithful.  Let $h$ be a nonzero element
	of $\hi$ and $P_h$ be the orthogonal projection onto the span of $h$.  
	Then $P_h$ is a positive non-zero operator on $\hi$.  Hence, since $\omega_{\rho}$ is
	faithful,
	\[
	0< \omega_{\rho}(P_h)=\Tr(\rho P_h)=\Tr( \rho^{1/2}P_hP_h \rho^{1/2})= 
	\| P_h\rho^{1/2} \|_2^2= 
	\frac{1}{\| h \|^2 } \| \rho^{1/2} h \|^2,
	\]
	and so $\rho^{1/2} h \neq 0$. By using the same argument with $h$ 
	replaced by $\rho^{1/2} h$, we have that $\rho h \neq 0$.  Thus, $\rho$ is injective.  
	
	$[(i) \Rightarrow (iii)]$.  Assume that $\rho$ does not have dense range and let
	$P$ be the orthogonal projection to $\Range (\rho)^{\perp}$. Then $P$ is a positive non-zero
	operator on $\hi$, and so 
	$\omega_{\rho}(P)>0$.  However, $P\rho=0$, and so 
	\[
	\omega_{\rho}(P)= \Tr( \rho P)=\Tr(P \rho)=\Tr(0)=0
	\]
	which is a contradiction. Thus, $\rho$ has dense range.
	
	$[(iii) \Rightarrow (i)]$.  Let $A \in \mathcal{B}(\hi)$ and suppose $\omega_{\rho}(A^*A)=0$.  Then
	\begin{equation}\label{equation1}
	0=\omega_{\rho}(A^*A)=\Tr (\rho A^*A)=\Tr(\rho^{1/2}A^*A \rho^{1/2})= \| A\rho^{1/2} \|_2^2.
	\end{equation}
	Hence $A \rho^{1/2}=0$, and therefore $A \rho=0$.  Since $\rho$ has dense range, this implies $A=0$. Thus, $\omega_{\rho}$ is faithful. 
	
	$[(ii) \Rightarrow (i)]$.  Assume that $\rho$ is injective and let $A \in \mathcal{B}(\hi)$ 
	such that $\omega_{\rho}(A^*A)=0$.  Equation \eqref{equation1} implies that 
	$A\rho^{1/2}=0$ and hence $\rho^{1/2}A^*=0$, thus $\rho A^*=0$. This implies $\rho A^*x=0$ 
	for any $x \in \hi$, and since $\rho$ is injective we have 
	that $A^*x=0$ for all $x\in \hi$, and so $A=0$. Thus, $\rho$ is faithful.
\qed\end{proof}

\begin{remark} \label{rmk1}
	Note that in the proof of $[(i)\Rightarrow (ii)]$ of the above proposition, we proved that (i) implies that $\rho^{1/2}$ is injective.  Since $\rho^{1/2}=\rho^{1/4}\rho^{1/4}$ we immediately obtain that $\rho^{1/4}$ is injective.  Since $\rho^{3/4}=\rho^{1/2}\rho^{1/4}$ we obtain that $\rho^{3/4}$ is injective as it is a composition of two injective maps. Further, since an operator is injective if and only if its adjoint has dense range, and $\rho^{1/4}$, $\rho^{1/2}$, and $\rho^{3/4}$ are self-adjoint, we have that $\rho^{1/4}$, $\rho^{1/2}$, and $\rho^{3/4}$ have dense range.
\end{remark}

\subsection{Inducing Maps on $\boldsymbol{S_2(\hi)}$}\label{inducing}
Let $\hi$ be a Hilbert space and fix $\rho \in \mathcal{S}_1(\hi)$ which is positive. Define 
\[
i_\rho: \bi (\hi) \rightarrow \bi (\hi )\quad \textrm{by} \quad i_\rho (x)= \rho^{1/4}x \rho^{1/4}.  
\]

The next proposition summarizes the properties of the map $i_\rho$. It is useful to first recall that for any Hilbert space $\hi$ the following set inclusions hold:
\[
\si_1 (\hi ) \subseteq \si_2 (\hi ) \subseteq \bi (\hi ) .
\]

\begin{proposition}\label{proposition2}
	Let $\rho \in \mathcal{S}_1(\hi)$ be positive such that $\omega_{\rho}$ is a faithful positive
	functional.  Then the following statements are valid:
	\begin{enumerate}
		\item[(a)] The map $i_\rho$ is injective.
		\item[(b)] The map $i_\rho$ preserves positivity.
		\item[(c)] The restriction $i_\rho |_{\si_2 (\hi )}$ of $i_\rho$ to $\si_2 (\hi )$ is a contraction from
		$\si_2 (\hi )$ into $\si_1 (\hi )$.
		\item[(d)] The map $i_\rho$ is a contraction from $\bi (\hi )$ onto a dense subset of $\si_2 (\hi )$.
	\end{enumerate}
\end{proposition}

\begin{proof}
	To prove (a), let $x \in \mathcal{B}(\hi)$ and suppose $i_\rho (x)=0$.  By 
	Remark~\ref{rmk1} 
	we have that $\rho^{1/4}$ is injective.
	Therefore, since  $\rho^{1/4}x\rho^{1/4}=0$, 
	we obtain that $x\rho^{1/4}=0$.  Further, since $\rho^{1/4}$ has dense range, 
	(by Remark~\ref{rmk1} again), 
	we obtain that $x=0$.  Thus $i_\rho$ is injective.
	
	To prove (b), let $x \in \mathcal{B}(\hi)$ where $x \geq 0$.  Let $h \in \hi$.  Then
	\[
	\la h,i_\rho (x)h \ra = \la h,\rho^{1/4}x \rho^{1/4}h \ra = \la  \rho^{1/4}h,x\rho^{1/4}h \ra
	\geq 0
	\]
	since $x \geq 0$. Thus $i_\rho$ maps positive operators to positive operators.  
	
	To prove (c), first note that for $p,q,r \geq 1$ with $\frac{1}{p}+\frac{1}{q} + \frac{1}{r}=1$ and for 
	$x \in \mathcal{S}_p(\hi)$, $y \in \mathcal{S}_q(\hi)$, and $z \in \mathcal{S}_r(\hi)$, two applications of 
	Holder's inequality give that $\| xyz \|_1 \leq \|x \|_p \| y \|_q \| z \|_r$.  From this we obtain that for $y \in \si_2 (\hi )$
	with $\| y \|_2 \leq 1$ we have
	\[
	\| i_\rho (y ) \|_1 = \| \rho^{1/4} y \rho^{1/4} \|_1 \leq \| \rho^{1/4} \|_4  \| y \|_2 \| \rho^{1/4} \|_4= 
	\| \rho \|_1^{1/4} \| y \|_2 \| \rho \|_1^{1/4}  \leq  \| \rho \|_1^{1/4}.\]
	
	To prove (d), first notice that $i_\rho (x) \in \si_2 (\hi)$ for all 
	$x \in \mathcal{B}(\hi)$ since
	\begin{align*}
	\|i_\rho (x) \|_2^2 &= \| \rho^{1/4} x \rho^{1/4} \|_2^2 
	=  \Tr\left( \rho^{1/2} x^* \rho^{1/2}x \right) 
	\\ &\leq \| \rho^{1/2} x^* \rho^{1/2}x \|_1 
	\leq \| \rho^{1/2} x^* \|_2 \| \rho^{1/2}x \|_2 
	\\ &= \Tr\left( (\rho^{1/2}x^*)^*( \rho^{1/2}x^*)\right)^{1/2}
	\Tr\left( (\rho^{1/2}x)^*(\rho^{1/2}x)\right)^{1/2}
	\\&= \Tr\left( x\rho x^*\right)^{1/2}\Tr\left( x^* \rho x \right)^{1/2}< \infty .
	\end{align*}
	
	Let $y \in \mathcal{S}_2(\hi)$ such that $y \perp i_\rho (x)$ for all $x \in \mathcal{B}(\hi)$,
	(where the orthogonality 
	is taken with respect to the Hilbert-Schmidt inner product).  Then, for all 
	$x \in \mathcal{B}(\hi)$, we have
	\begin{align*}
	0 &= \la i_\rho (x),y \ra_{\si_2(\hi )}= \la \rho^{1/4}x \rho^{1/4}, y \ra_{\si_2(\hi )} \\ &= 
	\Tr(\rho^{1/4}x^* \rho^{1/4}y )= 
	\la x , \rho^{1/4} y \rho^{1/4} \ra_{\si_2(\hi )}.
	\end{align*}
	Therefore $\rho^{1/4}y\rho^{1/4}=0$.  Since $\rho^{1/4}$ is injective, we have that 
	$y \rho^{1/4}=0$ and, since 
	$\rho^{1/4}$ has dense range, we have that $y=0$.  Therefore $i_\rho$ has dense range.
	
	To see that $\| i_\rho : \bi (\hi ) \to \si_2 (\hi ) \| \leq 1$, let $x \in \mathcal{B}(\hi)$ and 
	notice that
	\begin{align*}
	\| i_\rho (x) \|_2 & = \sup_{\substack{y \in \mathcal{S}_2(\hi) \\ \| y \|_2 \leq 1}} 
	| \la i_\rho (x) , y \ra_{\si_2(\hi )} |=
	\sup_{\substack{y \in \mathcal{S}_2(\hi) \\ \| y \|_2 \leq 1}} 
	| \Tr(i_\rho (x)^* y)| \\ 
	&=\sup_{\substack{y \in \mathcal{S}_2(\hi) \\ \| y \|_2 \leq 1}} 
	| \Tr(\rho^{1/4}y \rho^{1/4}x^*)| \leq 
	\sup_{\substack{y \in \mathcal{S}_2(\hi) \\ \| y \|_2 \leq 1}} 
	\| \rho^{1/4} y \rho^{1/4} \|_1 \| x \|_{\infty} \\
	& =  \| i_\rho |_{\si_2 (\hi )}:\si_2 (\hi ) \to \si_1 (\hi ) \| \| x \|_{\infty} \leq \| x \|_{\infty} ,
	\end{align*}
	where we used part (c) for the last inequality.
\qed\end{proof}

\begin{definition}
	Let  $\hi$ be a Hilbert space and  $\rho \in \si_1 (\hi )$ be a positive operator. 
	If $T: \mathcal{B}(\hi) \rightarrow \mathcal{B}(\hi)$ is a 
	bounded linear operator, we define the operator 
	$\widetilde{T}: i_\rho ( \mathcal{B}(\hi)) \rightarrow i_\rho (\mathcal{B}(\hi))$ by 
	\[
	\widetilde{T}( \rho^{1/4}x \rho^{1/4})= \rho^{1/4}T(x) \rho^{1/4} \quad \textrm{for all} \quad 
	x \in \mathcal{B}(\hi).
	\]
\end{definition}	
Note that $\widetilde{T}$ depends on $\rho$ but, for simplicity, we chose notation 
which does not reflect this dependence.

The following theorem was first proven in \cite{cf}.  For the convenience of the reader 
we provide a proof of it here.

\begin{theorem}\label{proposition3}
	Suppose $\hi$ is a Hilbert space and $\rho \in \si_1 (\hi )$ be a positive operator such that
	$\omega_\rho$ is a faithful positive functional on $\bi (\hi )$.
	Let $T: \mathcal{B}(\hi) \rightarrow \mathcal{B}(\hi)$ 
	be a bounded linear operator which is a  Schwarz map such that $\omega_\rho$ is a
	subinvariant functional for $T$. Then the corresponding operator $\widetilde{T}$ 
	can be extended to all of 
	$\mathcal{S}_2(\hi)$ as a contraction from $\si_2 (\hi )$ to $\si_2 (\hi )$.
\end{theorem}

\begin{proof}
	Since $\omega_\rho$ is a faithful normal functional on $\bi (\hi )$, we have that
	$\hi$ must be separable (see the comment above Proposition~\ref{proposition1}), so
	let $(e_k)_{k \geq 0}$ be an orthonormal basis for $\hi$ which diagonalizes $\rho$ and let 
	$P_n= \sum_{k=0}^n | e_k \ra \la e_k|$.  Note that $\rho$ and its positive powers commute
	with each $P_n$.  
	Define the linear subspace $\mathcal{A}= \{ x \rho^{1/2} : x \in \mathcal{B}(\hi) \}$ and 
	the map 
	$\widehat{T}:\mathcal{A} \rightarrow \mathcal{A}$ by 
	$\widehat{T}(x \rho^{1/2}) = T(x) \rho^{1/2}$.  
	Further, for $n \in \N$, define the map 
	$\Delta_n: \mathcal{S}_2(\hi ) \rightarrow \mathcal{S}_2(\hi)$ by
	\[
	\Delta_n(x)= P_n \rho^{1/2} x \rho^{-1/2}P_n \quad \textrm{ for all } \quad 
	x \in \mathcal{S}_2(\hi) 
	\]
	(note that $\rho^{1/2}$ is not invertible but $\rho^{-1/2}P_n$ is a bounded operator).  
	Then, for any $x\in \mathcal{B}(\hi)$, we have 
	\begin{align*}
	\| \widetilde{T}(i_\rho (x)) \|_2^2 & = \| \rho^{1/4}T(x)\rho^{1/4} \|_2^2 = 
	\Tr\left(\rho^{1/4}T(x )^* \rho^{1/2} T(x ) \rho^{1/4}\right)\\ 
	& = \lim_{n \rightarrow \infty}
	\Tr \left( \rho^{1/2}T(x)^*P_n \rho^{1/2} T(x) \rho^{1/2}\rho^{-1/2}P_n \right)
	\\& = \lim_{n \rightarrow \infty} 
	\left\la T(x )\rho^{1/2}, \Delta_n(T(x ) \rho^{1/2}) \right\ra_{\si_2(\hi )}\\ 
	& = \lim_{n \rightarrow \infty} 
	\left\la  \widehat{T}(x \rho^{1/2}), \Delta_n\widehat{T}(x \rho^{1/2}) \right\ra_{\si_2(\hi )}
	\\ &= \lim_{n \rightarrow \infty}
	\left\la x \rho^{1/2},\widehat{T}^* \Delta_n \widehat{T}( x \rho^{1/2}) \right\ra_{\si_2(\hi )} \stepcounter{equation}\tag{\theequation}\label{equationtheorem4}
	\end{align*}
	where we will see later on why $\widehat{T}^*$ is well-defined.
	
	Define $\Delta : \mathcal{A} \rightarrow \mathcal{A}$ by 
	$\Delta (x \rho^{1/2})=\rho^{1/2} x$, which is well-defined since 
	$\rho^{1/2}$ has dense range (hence, for $x,y \in \mathcal{B}(\hi)$, 
	$x\rho^{1/2}=y\rho^{1/2}$ implies $x=y$).  
	Let $\mathcal{B}= \{ x \rho : x \in \mathcal{B}(\hi) \}$.  We make the following three claims:
	\begin{itemize}
		\item[(i)] $\widehat{T}$ is a contraction on $\mathcal{A}$.
		Therefore $\widehat{T}$ can be extended to a contraction on $\mathcal{S}_2(\hi)$ since 
		$\mathcal{A}$ is dense in $\mathcal{S}_2(\hi)$). 
		\item[(ii)] $\Delta_n^2$ is positive.  Therefore, by  \cite[Lemma~1.2]{op}, 
		we have
		\begin{equation}\label{equation4.1}
		\widehat{T}^*\Delta_n \widehat{T} \leq \left( \widehat{T}^* \Delta_n^2 \widehat{T} \right)^{1/2}.
		\end{equation}
		\item[(iii)]  $\widehat{T}^* \Delta_n^2 \widehat{T} \leq \Delta^2$ on $\mathcal{B}$.  
		Thus,
		\begin{equation}\label{equation4.2}
		\left( \widehat{T}^* \Delta_n^2 \widehat{T} \right)^{1/2} \leq (\Delta^2)^{1/2}=\Delta .
		\end{equation}
		Hence, by combining \eqref{equation4.1} and \eqref{equation4.2}, we obtain 
		$\widehat{T}^* \Delta_n \widehat{T} \leq \Delta$ on $\mathcal{B}$.
	\end{itemize}
	Assume for the moment that the above claims (i), (ii), and (iii) are true. By replacing 
	$x$ by $x\rho^{1/2}$, in Equation \eqref{equationtheorem4}
	we obtain that
	\begin{align*}
	\| \widetilde{T}(i_\rho (x\rho^{1/2})) \|_2^2 
	& = \lim_{n \rightarrow \infty} 
	\left\la x \rho, \widehat{T}^* \Delta_n \widehat{T}(x \rho) \right\ra_{\si_2(\hi )} 
	\\ &\leq \left\la x\rho, \Delta (x\rho) \right\ra_{\si_2(\hi )}  
	= \left\la x\rho, \rho^{1/2} x \rho^{1/2} \right\ra_{\si_2(\hi )} \\ 
	&= \Tr\left( \rho x^* \rho^{1/2}x \rho^{1/2} \right)  
	= \Tr \left( \rho^{3/4}x^* \rho^{1/4}\rho^{1/4}x \rho^{3/4} \right) \\
	& = \left\la \rho^{1/4} x \rho^{3/4}, \rho^{1/4} x \rho^{3/4} \right\ra_{\si_2(\hi )}
	= \| i_\rho (x\rho^{1/2}) \|_2^2
	\end{align*}
	and so $\widetilde{T}$ is a contraction on $i_\rho (\mathcal{B}(\hi)\rho^{1/2} )$.  
	We now show that 
	$i_\rho (\mathcal{B}(\hi) \rho^{1/2})$ is dense in $\mathcal{S}_2(\hi)$.  Let 
	$y \in \mathcal{S}_2(\hi)$ such that 
	$y \perp i_\rho (\mathcal{B}(\hi)\rho^{1/2})$.  Then, for any $x \in \mathcal{B}(\hi)$ we have that
	\begin{align*}
	0 &= \la i_\rho (x \rho^{1/2}), y \ra_{\si_2(\hi )}=
	\Tr(i_\rho (x\rho^{1/2})^*y ) \\&=
	\Tr(\rho^{1/4}\rho^{1/2}x^* \rho^{1/4}y)
	=\la x, \rho^{1/4}y \rho^{3/4} \ra_{\si_2(\hi )}
	\end{align*}
	and hence $\rho^{1/4}y\rho^{3/4}=0$.  Since $\rho^{1/4}$ is injective, we then have that 
	$y \rho^{3/4}=0$ and, since 
	$\rho^{3/4}$ has dense range, we obtain that $y=0$.  Therefore 
	$i_\rho (\mathcal{B}(\hi) \rho^{1/2})$ is dense in 
	$\mathcal{S}_2(\hi)$.  Since $\widetilde{T}$ is a contraction on 
	$i_\rho (\mathcal{B}(\hi) \rho^{1/2})$, we can extend it to a 
	contraction on $\mathcal{S}_2(\hi)$.  This finishes the proof of the theorem pending verification of claims (i), (ii), and (iii), as well as the fact that $\widehat{T}^*$ is 
	well-defined.
	
	First, we prove claim (i), i.e., that $\widehat{T}$ is a contraction on $\mathcal{A}$.  Let $x \in \mathcal{B}(\hi)$.  
	Then 
	\begin{align*}
	\| \widehat{T}(x \rho^{1/2}) \|_2^2 &= \| T(x) \rho^{1/2} \|_2^2 = 
	\la T(x) \rho^{1/2} , T(x) \rho^{1/2} \ra_{\si_2(\hi )} \\ &= 
	\Tr( \rho^{1/2} T(x)^*T(x) \rho^{1/2} )  \leq \Tr( \rho^{1/2} T(x^*x) \rho^{1/2} )
	\end{align*}
	since $T$ is a Schwarz map.  Further,
	\begin{align*}
	\Tr( \rho^{1/2} T(x^*x) \rho^{1/2} ) &= \Tr( \rho T(x^*x)) \leq \Tr( \rho x^*x) = 
	\Tr( \rho^{1/2} x^*x \rho^{1/2} ) \\
	&= \la x\rho^{1/2} , x \rho^{1/2} \ra_{\si_2(\hi )} = \| x \rho^{1/2} \|_2^2.
	\end{align*}
	Therefore $\| \widehat{T}(x \rho^{1/2}) \|_2^2 \leq \| x \rho^{1/2} \|_2^2$, and so 
	$\widehat{T}$ is a contraction on 
	$\mathcal{A}$.  Hence, $\widehat{T}$ can be extended to a contraction on 
	$\mathcal{S}_2(\hi)$ since 
	$\mathcal{A}$ is dense in $\mathcal{S}_2(\hi)$ (this also shows that 
	$\widehat{T}^*$ is well-defined).
	
	For claim (ii), i.e., that $\Delta_n^2$ is positive, first note that since $\rho$ commutes 
	with $P_n$ we have
	\[
	\Delta_n^2x= P_n \rho x P_n \rho^{-1}P_n \quad \textrm{ for all } \quad 
	x \in \mathcal{S}_2(\hi) 
	\]
	(note that $\rho$ is not invertible but $\rho^{-1}P_n$ is a bounded operator).  Indeed, if 
	$x \in \mathcal{S}_2(\hi)$ then
	\begin{align*}
	\left\la  x , \Delta_n^2x \right\ra_{\si_2(\hi )} & = 
	\left\la x, P_n \rho x P_n \rho^{-1}P_n \right\ra_{\si_2(\hi )}
	= \Tr\left( x^* P_n \rho x P_n \rho^{-1} P_n  \right) \\ 
	& = \Tr\left( \rho^{1/2} P_n x P_n\rho^{-1/2}P_nP_n \rho^{-1/2}P_nx^* P_n \rho^{1/2}  \right)\\ 
	& = \Tr\left( ( \rho^{1/2} P_nxP_n\rho^{-1/2}P_n)( \rho^{1/2} P_nxP_n\rho^{-1/2}P_n)^* \right) \geq 0,
	\end{align*}
	and so $\Delta_n^2$ is positive.  By \cite[Lemma 1.2]{op}, we then have that 
	\begin{equation}\label{equation2}
	\widehat{T}^*\Delta_n \widehat{T} \leq (\widehat{T}^* \Delta_n^2 \widehat{T})^{1/2}.
	\end{equation}
	
	It is left to prove claim (iii), i.e., that $\widehat{T}^* \Delta_n^2 \widehat{T} \leq \Delta^2$
	on $\mathcal{B}$.  Indeed,
	\begin{align} \label{E:seebelow}
	\left\la x\rho, \widehat{T}^* \Delta_n^2\widehat{T}(x \rho) \right\ra_{\si_2(\hi )} 
	& = 
	\left\la 
	T(x \rho^{1/2})\rho^{1/2}, \Delta_n^2T(x \rho^{1/2})\rho^{1/2}  
	\right\ra_{\si_2(\hi )} \nonumber \\ 
	& = \left\la  
	T(x \rho^{1/2})\rho^{1/2}, P_n \rho T(x \rho^{1/2})\rho^{1/2} P_n \rho^{-1}P_n 
	\right\ra_{\si_2(\hi )} \nonumber \\ 
	& =\Tr\left(
	\rho^{1/2} T(x\rho^{1/2})^* P_n \rho T(x \rho^{1/2}) \rho^{1/2}P_n \rho^{-1}P_n 
	\right) \nonumber \\
	& = \Tr\left( 
	\rho T(x \rho^{1/2})P_n T(x \rho^{1/2})^*P_n
	\right) \nonumber \\ 
	& \leq \Tr\left(\rho T(x \rho^{1/2})T(x \rho^{1/2})^*\right) \quad \textrm{(see below)}\\ 
	& \leq \Tr\left( \rho T\left((x \rho^{1/2})(x \rho^{1/2})^*\right) \right) \quad (T \textrm{ is a Schwarz map)}\nonumber \\
	& \leq  \Tr\left( \rho (x \rho^{1/2})(x \rho^{1/2})^*\right) \quad (\omega_\rho \textrm{ is subinvariant for } T) \nonumber\\ 
	&= \Tr\left(\rho x\rho x^*\right) =\Tr\left(\rho x(x\rho)^*\right) = \Tr\left((x\rho)^*\rho x\right)  \nonumber\\
	& = \Tr\left((x \rho )^* \Delta^2(x\rho)\right)\quad 
	\textrm{(since }\Delta^2 (x\rho)= \rho x ) \nonumber
	\\ &= \left\la x \rho, \Delta^2(x \rho) \right\ra_{\si_2(\hi )} . \nonumber
	\end{align}
	This completes the proof as long as we justify the inequality (\ref{E:seebelow}).  
	Indeed, we have that the inequality $\Tr(P_nA^*P_nA) \leq \Tr(A^*A)$ holds in general for any 
	$A \in \mathcal{S}_2(\hi)$, since if $(e_k)_{k \geq 1}$ is the 
	orthonormal basis of $\hi$ used to define each $P_n$, then
	\begin{align*}
	\Tr(P_nA^*P_nA)  &= \sum_{k=1}^{\infty} \left\la e_k, P_nA^*P_n^2Ae_k \right\ra  
	\\ &= \sum_{k=1}^{\infty} \left\la P_nAP_ne_k, P_nAe_k \right\ra =\sum_{k=1}^n \left\la P_nAe_k, P_nAe_k \right\ra,
	\end{align*}
	and further
	\begin{align*}
	\sum_{k=1}^n \left\la P_nAe_k, P_nAe_k \right\ra &=\sum_{k=1}^n \| P_nAe_k \|^2 
	\leq \sum_{k=1}^n \| P_n \|^2 \|Ae_k \|^2\\ &\leq \sum_{k=1}^{\infty} \|Ae_k \|^2 =\Tr(A^*A).\end{align*}
\qed\end{proof}

\subsection{An Alternate Construction}\label{alternate}

There is another situation where a bounded operator on a $C^*$-algebra gives rise to
a corresponding operator on a Hilbert space, and we would like to mention this 
in the next remark.

\begin{remark} \label{R:kfgv}
	Let $\ai$ be a unital $C^*$-algebra and $\omega$ be a faithful state on $\ai$.
	Consider the GNS construction of $\ai$ associated with
	$\omega$. Let $\ki$ be the Hilbert space associated with the GNS construction, 
	$\pi : \ai \to \bi (\ki)$ be the $*$-representation of $\ai$ into the $C^*$-algebra of all
	bounded operators on $\ki$, and $\Omega$ denote the cyclic element of the Hilbert
	space $\ki$ for the representation $\pi$, 
	(i.e. the subspace $\{ \pi (a) ( \Omega ) : a \in \ai \}$ is norm dense in $\ki$) which is 
	equal to the unit of $\ai$ viewed as an element of $\ki$. 
	Let $T$ be a  bounded operator on $\ai$ which is  a Schwarz map. 
	Assume that $\omega$ is subinvariant for $T$.
	Define  an operator 
	$\overline{T}$ on the dense subspace $\{ \pi (a)(\Omega ) : a \in \ai \}$ of $\ki$
	with values in $\ki$ by 
	\[
	\overline{T}(\pi (a) (\Omega)) = \pi (T(a)) (\Omega ) \quad \textrm{for all }a \in \ai.
	\]
	Then $\overline{T}$ is a contraction (hence it extends to $\ki$).
\end{remark}

\begin{proof}
	Since $\omega$ is faithful, the quotient that is usually 
	associated with the 
	GNS construction does not take place, and the elements of $\ai$ belong to $\ki$.
	Let $\la \cdot , \cdot \ra_\omega$ denote the inner product in $\ki$ and 
	$\| \cdot \|_\omega$ denote the norm of $\ki$. Then since $\omega$ is faithful,
	we have that for  $a,b \in \ai$, 
	$\la a , b \ra_\omega = \omega (a^* b)$
	and hence $\| (\pi(a))(\Omega) \|_\omega^2 = \omega (a^*a)$. 
	
	For every $a \in \ai$ we have
	\begin{align*}
	\| \overline{T} (\pi (a) (\Omega)) \|_\omega ^2  = & \| \pi (T(a)) (\Omega ) \|_\omega^2
	= \omega (T(a)^* T(a))  \leq \omega (T(a^*a)) \\
	& \textrm{(since }\omega \textrm{ is positive and }T \textrm{ is a Schwarz map)}\\
	\leq &\omega (a^* a ) \quad
	\textrm{($T\geq 0$ is a Schwarz map; $\omega$ is subinvariant for $T$)}\\
	= &\| \pi (a) (\Omega) \|_\omega^2,
	\end{align*}
	which finishes the proof.
\qed\end{proof}

Notice the similarities between Theorem~\ref{proposition3} and Remark~\ref{R:kfgv}. 
Both refer to a bounded operator on some $C^*$-algebra  where a positive linear 
functional is fixed, and they each conclude the 
existence of an associated contraction on some Hilbert space. But there are three key differences
between Theorem~\ref{proposition3} and Remark~\ref{R:kfgv}. First, 
Theorem~\ref{proposition3} refers to an operator on 
$\bi (\hi )$ for some Hilbert space $\hi$ (which is necessarily separable since $\bi (\hi)$
is assumed to admit a faithful normal state), while Remark~\ref{R:kfgv} assumes that 
the operator is defined on a general $C^*$-algebra. Second, the state $\omega_\rho$
which is mentioned in Theorem~\ref{proposition3} is normal since it is defined via 
the trace-class operator $\rho$, while there is no such assumption 
in Remark~\ref{R:kfgv} (the normality of the 
state $\omega$ in Remark~\ref{R:kfgv} does not make
sense in general  since $\ai$ is simply 
assumed to be a $C^*$-algebra and not a von~Neumann algebra as it is assumed in \cite[Equation (2.1)]{kfgv}). 
Third, the Hilbert space that is used in 
Theorem~\ref{proposition3} is the space $\si_2 (\hi )$ 
which does not depend on the positive linear functional,
while the map  $i_\rho$ which maps $\bi (\hi )$ to $\si_2 (\hi )$, does depend on the 
positive linear functional.
On the other hand, 
the Hilbert space that is used in Remark~\ref{R:kfgv}  
(i.e. the GNS construction
associated to the 
faithful state $\omega$ of the $C^*$-algebra $\ai$)
depends on the state, while the $*$-representation $\pi$ of the von~Neumann algebra 
which is associated with the GNS construction does not depend on the state.   
Notice also that the combinations of the Hilbert spaces with the representations  in 
Theorem~\ref{proposition3} and Remark~\ref{R:kfgv} are very similar. 
More precisely, for $a, b \in \bi (\hi )$ we have that 
$i_\rho (a), i_\rho (b) \in \si_2 (\hi)$ hence 
\begin{align*}
\la i_\rho (a), i_\rho (b) \ra_{\si_2 (\hi ) } &= \Tr (i_\rho (a)^* i_\rho (b) ) \\&= 
\Tr (\rho^{1/4} a^* \rho^{1/4} \rho^{1/4} b \rho^{1/4} ) =
\Tr (a^* \rho^{1/2} b \rho^{1/2}) . 
\end{align*}
On the other hand, if we assume for the moment that 
the $C^*$-algebra  $\ai$ that appears in Remark~\ref{R:kfgv} is equal to $\bi (\hi )$ for
some Hilbert space $\hi$, and 
the faithful state $\omega$ on the $C^*$-algebra 
$\ai$ is given by $\omega (a) = \Tr (\rho a)$ for some positive trace-class operator 
$\rho$ on $\hi$, then the inner product of two elements
$a,b \in \ai$ via the GNS construction is given by
\[
\la a, b \ra_{\omega} =\omega (a^*b) = \Tr (\rho a^* b) .
\]
Thus the  combination of the inner product with the representation that is used 
in Theorem~\ref{proposition3} is slightly more ``symmetric" than the combination of the
inner product with the representation that is used in Remark~\ref{R:kfgv}. 
The reader of course will notice the difference between the complexity of the 
proof of Theorem~\ref{proposition3}  and that of  Remark~\ref{R:kfgv}. 
The extra intricacies in the proof of Theorem~\ref{proposition3} is the price we pay 
in order to achieve the extra symmetry 
in the combination of the inner product and the representation as discussed above.

\begin{remark}
	The assumption that ``$\omega$ is subinvariant for $T$" cannot be 
	omitted in Remark~\ref{R:kfgv}.
\end{remark}

An example where $\omega$ is not a subinvariant functional for $T$ but all the other assumptions of Remark~\ref{R:kfgv} are valid is presented in Remark 11.7 of \cite{WDiss}. 

\begin{remark}
	Note that if $\hi$ is a Hilbert space, $T:\bi (\hi) \to \bi (\hi)$ is a bounded positive linear
	operator, and $\omega$ is a subinvariant positive faithful functional for $T$,
	then  $\omega/\omega(1)$ is a subinvariant faithful state for $T$ (here $1$ denotes the 
	identity operator on $\hi$). Thus, instead of assuming the existence of 
	subinvariant positive faithful functionals, we henceforth simply assume the existence 
	of subinvariant faithful states. Our subsequent results thus
	remain valid if 
	the assumptions of the existence of subinvariant faithful states are replaced by the assumptions of the existence of subinvariant positive faithful functionals.
\end{remark}

\section{Semigroups of Schwarz Maps}\label{schwarz}

We first recall some basic definitions about semigroups.

\begin{definition} \label{Def:semigroupcontinuity}
	Let $X$ be a Banach space.  A one-parameter family $(T_t)_{t\geq 0}$ of bounded operators on $X$ is a 
	\textbf{semigroup} on $X$ if  $T_{t+s}=T_tT_s$ for all $t,s \geq 0$, and
	$T_0=I$ where $I$ is the identity operator on $X$.
	We say the semigroup $(T_t)_{t \geq 0}$ on a Banach space $X$  is
	\begin{itemize}
		
		\item \textbf{uniformly continuous} if the map $t \mapsto T_t$ is continuous with 
		respect to the operator norm. 
		
		\item \textbf{strongly continuous} if for all 
		$x \in X$ the map $t \mapsto T_tx$ is continuous with respect to the norm on $X$.
		
		\item \textbf{weakly continuous} if for all $x \in X$ and all $x^* \in X^*$ the map 
		$t \mapsto x^*(T_tx)$ is continuous.  
		
		\item \textbf{weak}$\mathbf{^*}$~\textbf{continuous} if $X$ is a dual Banach space $X= Y^*$ and for all 
		$y \in Y$ and $x \in X$ the map $t \mapsto (T_t(x))(y)$ is continuous.
		
	\end{itemize}
	If $\hi$ is a Hilbert space and $X=\bi (\hi)$ then the semigroup 
	$(T_t)_{t \geq 0}$ on the Banach space $X$  is  \textbf{WOT continuous} (where this acronym stands 
	as usually for the weak operator topology) if for all $h_1,h_2 \in \hi$
	and $x \in \bi (\hi )$ we have that 
	the map $t \mapsto \la h_1, T_t(x)h_2 \ra $ is continuous.
\end{definition}

It can be shown that a semigroup on a Banach space is strongly continuous 
if and only if it is weakly continuous (see \cite[Thm. 3.31]{attal}). If $(T_t)_{t \geq 0}$ is a uniformly
continuous semigroup on a Banach space $X$ then its \textbf{generator} is defined as 
the operator norm limit
\[
L= \lim_{t \to 0} \frac{T_t - I}{t}.
\]
This limit exists and it defines a bounded operator on $X$. If we do not
assume the uniform continuity of the semigroup, then the definition of the generator is
given next:

\begin{definition}
	Let $(T_t)_{t \geq 0}$ be a strongly continuous semigroup (resp. weakly continuous, resp. weak$^*$ continuous), 
	on a Banach space $X$ (of course, when we assume that the semigroup is 
	weak$^*$ continuous we assume that 
	$X$ is a dual Banach space).  We say an element $x \in X$ belongs to the 
	\textbf{domain} $D(L)$ of the generator $L$ of $(T_t)_{t \geq 0}$,  if
	\begin{equation}\label{defgen1}
	\lim_{t \rightarrow 0} \frac{T_t(x) -x}{t} 
	\end{equation}
	converges in norm (resp. weakly, resp. weak$^*$) and, in this case, define the 
	\textbf{generator} to be the generally unbounded operator $L$ such that
	\begin{equation}\label{defgen2}
	L(x)=\lim_{t \rightarrow 0} \frac{T_t(x) -x}{t}  \quad \textrm{for all } x \in D(L)
	\end{equation}
	where the last limit is taken in the norm (resp. weak, resp. weak$^*$) topology of $X$. 
\end{definition}

Since a semigroup on a Banach space is strongly continuous if and only if it is weakly continuous, it is natural to 
ask whether the limits \eqref{defgen1} and \eqref{defgen2} can be replaced by weak limits and end up with the same 
$D(L)$ and $L$.  It turns out that this is indeed the case (see \cite[Proposition 3.36]{attal}).  We will make use of this fact in the proof of Theorem~\ref{theorem4}.

\subsection{The Extended Generator $\boldsymbol{L_{(h_n)}}$ of $\boldsymbol{(T_t)_{t\geq0}}$}\label{extending}

We now wish to extend the definition of the generator to include some cases where the limit \eqref{defgen2} does not exist. We first require the following notation:

\begin{definition}
	Let $\hi$ be a Hilbert space and $(h_n)_{n \in N}$ be an (countable or uncountable) orthonormal basis of $\hi$. We let $M_N^{(h_n)}$ denote the set of all complex $N\times N$ matrices with rows and columns indexed by $N$. We view a matrix $L\in M_N^{(h_n)}$ as a linear map $L:D(L)\to \C^N$ acting on $\hi$ as follows: denote $L=(L_{n,m})_{n,m\in\ N}$, and define $D(L)\subset \hi$ as the set of all vectors $h=\sum_{m\in N}\la h_m,h\ra h_m\in\hi$ such that the series $\sum_{m\in N}L_{n,m}\la h_m,h\ra$ converges for all $n\in N$. Then \[L(h)=\left(\sum_{m\in N}L_{n,m}\la h_m,h\ra\right)_{n\in N}.\] This is in particular the natural matrix multiplication of $L$ against $h$ written as a column vector.
\end{definition}

The following definition is given as the minimal requirements for the outputs of $L$ to be considered as linear maps in the sense given above for a fixed orthonormal basis $(h_n)_{n\in N}$ of~$\hi$.

\begin{definition} \label{Def:generatorwrtabasis}
	Let $\hi$ be a Hilbert space and $(h_n)_{n \in N}$ be a 
	(countable or uncountable) orthonormal basis of $\hi$. Let
	$(T_t)_{t \geq 0}$ be a semigroup 
	of bounded operators on $\bi (\hi )$. To define the extended generator $L_{(h_n)}$
	of $(T_t)_{t \geq 0}$ with respect to the basis $(h_n)_{n \in N}$ we first define its domain as the linear subspace of all $x \in \bi (\hi )$ such 
	that the function
	\[
	[0,\infty ) \ni t \mapsto \la h_n, T_t(x) h_m \ra 
	\]
	is  differentiable at $0$ for every $n,m \in N$; that is, $D(L_{(h_n)})$ is the linear subspace of all $x \in \bi (\hi )$ such that the limit
	\[
	\lim_{t \to 0} \la h_n, \frac{T_t(x)-x}{t}h_m \ra
	\]
	exists for every $n,m \in N$. In general $D(L_{(h_n)})$ can be 
	the zero subspace, but if the semigroup is WOT continuous then
	$D(L_{(h_n)})$ is WOT dense in $\bi (\hi)$. Define the 
	\textbf{extended generator} $\mathbf{L_{(h_n)}}$ \textbf{of} 
	$\mathbf{(T_t)_{t \geq 0}}$ \textbf{(with respect to the orthonormal basis} 
	$\mathbf{(h_n)_{n \in N}})$ to be the map with domain
	$D(L_{(h_n)})$ whose range elements are matrices $L_{(h_i)}(x)\in M_N^{(h_n)}$ with entries given by \[[L_{(h_i)}(x)]_{n,m}=\lim_{t \to 0} \la h_n, \frac{T_t(x)-x}{t}h_m \ra.\]
\end{definition}

Next we want to compare the generator of a semigroup on
$\bi (\hi )$ with respect to an orthonormal basis of $\hi$ to 
the usual generator. Since the definition of the generator depends
on the continuity of the semigroup, 
in the next remark we will consider a weak$^*$ continuous semigroup
on $\bi (\hi)$ for some Hilbert space $\hi$. 
The reason that we choose the weak$^*$ continuity versus any
other continuity assumption is because it is the weakest and the most natural
among all continuity assumptions that appear in 
Definition~\ref{Def:semigroupcontinuity}.

\begin{remark}\label{rem1}
	Let $\hi$ be a Hilbert space,
	$(T_t)_{t \geq 0}$ be a weak$^*$ continuous 
	semigroup of bounded operators on $\bi (\hi )$, and let $L$ denote its generator. Let $(h_n)_{n \in N}$ be a (countable or uncountable) 
	orthonormal basis of $\hi$, and let $L_{(h_n)}$ denote the 
	generator of $(T_t)_{t \geq 0}$ with respect to $(h_n)_{n \in N}$.
	Then $D(L) \subseteq D(L_{(h_n)})$,
	and for every $x \in D(L)$ we have	$L(x)=L_{(h_n)}(x)$, by which we mean the matrix of $L(x)$ with respect to $(h_n)_{n\in N}$ and the matrix $L_{(h_n)}(x)$ are equal.
\end{remark}

Indeed, for fixed $x \in D(L)$ and every $h,h' \in \hi$ we have that 
\begin{equation} \label{E:deriv}
\la h,  \frac{T_t(x) -x}{t}h' \ra \to \la h,  L(x)  h' \ra 
\quad \textrm{as } t \to 0.
\end{equation}
In particular,
\[
\lim_{t \to 0}\limits \la h_n , \frac{T_t(x)-x}{t}h_m \ra =\la h_n,L(x)h_m\ra
\]
for every $n,m \in N$. Thus $x \in D(L_{(h_n)})$.

\noindent \textbf{Notation:} If $N$ is a nonempty set, then we denote by $\Pi_\textrm{fin}(N)$ the set of all finite subsets of $N$.

\noindent \textbf{Notation:}
	Let $\hi$ and $\ki$ be Hilbert spaces with $\hi \subseteq \ki$ and let 
	$A \in \bi (\hi)$ and $B \in \bi (\ki)$.  We shall denote by
	\[
	A=\pr_\hi (B)
	\]
	the fact that 
	\[
	A=P_\hi B|_\hi
	\]
	where $\phantom{A}|_\hi$ denotes the restriction to $\hi$
	and $P_\hi :\ki \to \hi$ denotes the orthogonal projection from $\ki$ onto $\hi$.
	The operator $B$ is called a \textbf{dilation} of the operator $A$ and the operator 
	$A$ is called a \textbf{compression} of the operator $B$.

\begin{remark}\label{rem2}
	Let $\hi$ be a Hilbert space with (countable or uncountable) dimension $N$, $(h_n)_{n \in N}$ be an orthonormal basis of $\hi$,
	$(T_t)_{t \geq 0}$ be a 
	semigroup of bounded operators on $\bi (\hi )$, and let 
	$L_{(h_n)}$ denote its generator with respect to $(h_n)_{n \in N}$. 
	For $x \in D(L_{(h_n)})$ and $F \in \Pi_\textrm{fin}(N)$ there exists a unique operator
	\[
	L_{(h_n)}(x)_F : \Span (h_n)_{n \in F} \to \Span (h_n)_{n \in F}
	\]
	satisfying \begin{equation} \label{E:generatorwrtabasis2}
	\lim_{t\to 0}\limits \la h , \frac{T_t(x)-x}{t} h' \ra = 
	\la h , L_{(h_n)}(x)_F h' \ra \quad \textrm{for all }
	h, h' \in \Span(h_n)_{n \in F},
	\end{equation} or equivalently 
	\begin{equation} \label{E:ED}
	\left\| 
	\pr_{\Span (h_n)_{n \in F}}\left( \frac{T_t (x)-x}{t}\right) - 
	L_{(h_n)}(x)_F
	\right\|_{\bi (\Span (h_n)_{n \in F})} \to 0 \quad \textrm{as }t \to 0.
	\end{equation}
\end{remark}

Indeed, fix $F \in \Pi_\textrm{fin}(N)$. From 
Definition~\ref{Def:generatorwrtabasis}, $L_{(h_n)}(x)_F: \Span(x_n)_{n \in F} \to \Span(x_n)_{n \in F}$ is uniquely defined by
\begin{equation} \label{E:generatorwrtabasis2'}
[L_{(h_i)}(x)_F](h)=\sum_{n,m\in F}\lim_{t\to0}\la h_n,\frac{T_t(x)-x}{t}h_m\ra\la h_m,h \ra h_n\quad\textrm{if } h=\sum_{m\in F}\la h_m,h\ra h_m.
\end{equation}

Then (\ref{E:generatorwrtabasis2}) is obvious from Definition~\ref{Def:generatorwrtabasis} and \eqref{E:generatorwrtabasis2'}. The equivalence of \eqref{E:generatorwrtabasis2} and \eqref{E:ED} then follows for any finite subset $F$ of $N$, since all linear Hausdorff topologies 
on the space of linear operators on $\Span (h_n)_{n \in F}$ are equivalent. Thus the WOT on $\Span (h_n)_{n \in F}$ in 
(\ref{E:generatorwrtabasis2}) can be replaced by the
$\bi (\Span (h_n)_{n \in F})$ topology.

\begin{remark} \label{compatible}
	Let $\hi$ be a Hilbert space with (countable or uncountable) dimension $N$, $(h_n)_{n \in N}$ be an orthonormal basis of 
	$\hi$, 
	$(T_t)_{t \geq 0}$ be a 
	semigroup of bounded operators on $\bi (\hi )$, and let 
	$L_{(h_n)_n}$ denote its generator with respect to 
	$(h_n)_{n \in N}$. 
	Fix $x \in D (L_{(h_n)})$. Then the family 
	$(L_{(h_n)}(x)_F)_{F \in \Pi_\textrm{fin}(N)}$ is compatible in the following sense:
	If $G \subset F$ are two finite subsets of $N$ then 
	$\textrm{pr}_{\Span (h_n)_{n \in G}}(L_{(h_n)}(x)_F)=L_{(h_n)}(x)_G$.
\end{remark}

Indeed, this is obvious from \eqref{E:generatorwrtabasis2'}.

\begin{remark}
	The generator of a semigroup with respect to an orthonormal 
	basis that we defined above is related to the 
	\textbf{form generator} which was defined by Davies \cite{D77}
	and was further studied in \cite{holevo93}, \cite{CF93}, \cite{sinha}, \cite{holevo}, \cite{CF98}, \cite{fagnola1}, \cite{az}, and \cite{hsw}. If 
	$(T_t)_{t \geq 0}$ is a weak$^*$ continuous semigroup on 
	$\bi (\hi )$ for some Hilbert space $\hi$, then a form generator 
	is the map $\phi : \ki \times \bi (\hi ) \times \ki \to \C$
	where $\ki$ is a dense linear subspace of $\hi$, defined by
	\[
	\phi (h, x , h')= \la h, \lim_{t \to 0}\limits \frac{T_t(x)-x}{t}h' \ra 
	\quad \textrm{ for every }h,h' \in \ki \textrm{ and every }x \in \bi (\hi ) .
	\] Note that if $(h_n)_{n \in N}$ is an orthonormal basis of $\hi$
	and $\ki$ denotes the linear span of $(h_n)_{n \in N}$
	then the form generator coincides with the generator with respect
	to $(h_n)_{n \in N}$ if the domain of the generator with respect to 
	$(h_n)_{n \in N}$ is equal to $\bi (\hi )$. Here we assume that
	the domain of the generator with respect to an orthonormal basis 
	is a linear subspace of $\bi (\hi )$, not necessarily equal to 
	$\bi (\hi )$. 
\end{remark}

We require a few more definitions in order to state the next result.

\begin{definition}
	Let $\hi$ be a Hilbert space, $\omega$ be a state on $\bi (\hi )$ and $(T_t)_{t \geq 0}$
	be a semigroup of positive operators on $\bi (\hi )$. We say that 
	$\omega$ is a 
	subinvariant state for the semigroup $(T_t)_{t \geq 0}$, if and only if $\omega$ is 
	subinvariant for $T_t$ for every $t \geq 0$.
\end{definition}

\begin{definition}The \textbf{Moore-Penrose inverse} or \textbf{pseudoinverse} $x^{(-1)}$ of $x\in \bi(\hi)$ is defined as the unique linear extension of $(x|_{\Ni(x)^\perp})^{-1}$, the inverse as a function, to \[D(x^{(-1)}):=\Ri(x)+\Ri(x)^\perp\] with $\Ni(x^{(-1)})=\Ri(x)^\perp$, where $\Ni(x)$ and $\Ri(x)$ denote the nullspace and range of $x$, respectively. Letting $P$ and $Q$ denote the orthogonal projections onto $\Ni(x)$ and $\overline{\Ri(x)}$, respectively, it can be shown (see e.g. \cite{EnglHanke}) that $x^{(-1)}$ is uniquely determined by the relations \[x^{(-1)}x=I-P \quad \textrm{and}\quad xx^{(-1)}=Q|_{D(x^{(-1)})}.\] \end{definition}

\noindent \textbf{Notation:} By $i_{\rho^{(-1)}}$ we mean the map from $\bi(\hi)$ to the space of linear maps on $\hi$ defined via \[i_{\rho^{(-1)}}(x)=(\rho^{1/4})^{(-1)}x(\rho^{1/4})^{(-1)}.\]

Now we are ready to state the next result.

\begin{theorem}\label{theorem4}
	Let $\hi$ be a Hilbert space, $(T_t)_{t\geq 0}$ be a semigroup of Schwarz maps on 
	$\bi (\hi)$, and let 
	$\rho \in \mathcal{S}_1(\hi)$ be such that $\omega_{\rho}$ is a faithful  state
	on $\bi (\hi)$ which is subinvariant for the semigroup  
	$(T_t)_{t \geq 0}$.
	Then there exists a unique semigroup 
	$(\widetilde{T_t})_{t \geq 0}$ of contractions on 
	$\mathcal{S}_2(\hi)$ such that
	\begin{equation}\label{equation4}
	\widetilde{T_t}(i_\rho (x))=i_\rho (T_t(x)) \quad \textrm{ for all }  x \in \mathcal{B}(\hi).
	\end{equation}
	Moreover, if $(T_t)_{t \geq 0}$ is weak$^*$-continuous then $(\widetilde{T_t})_{t \geq 0}$ 
	is strongly continuous.  
	Let  $L$ denote the generator of $(T_t)_{t \geq 0}$, let $\widetilde{L}$ denote the generator of 
	$(\widetilde{T_t})_{t \geq 0}$, and let $L_{(h_n)}$ denote the generator of $(T_t)_{t \geq 0}$ with 
	respect to $(h_n)_{n \in \N}$, where $(h_n)_{n \in \N}$ is an orthonormal basis of $\hi$ consisting of eigenvectors of $\rho$. Then $x \in D(L)$ implies $i_\rho(x) \in D( \widetilde{L})$, and moreover \[\widetilde{L}(i_\rho(x))=i_\rho(L(x));\] conversely, 
	$i_\rho(x) \in D( \widetilde{L})$ implies $x \in D(L_{(h_n)})$,
	and moreover \begin{equation}
	\label{LtotildeL} L_{(h_n)}(x)=i_{\rho^{(-1)}}(\widetilde{L}(i_\rho(x))).\end{equation}
\end{theorem}

\begin{proof}
	The operators $\widetilde{T_t}$ are well-defined by Theorem \ref{proposition3}. 
	Uniqueness comes from 
	Equation \eqref{equation4} and the fact that $i_\rho (\mathcal{B}(\hi))$ is dense in 
	$\mathcal{S}_2(\hi)$.  It is easy to see that 
	$\widetilde{T}_{t+s}=\widetilde{T_t}\widetilde{T}_s$ and that $\widetilde{T}_0=1$ on 
	$i_\rho (\mathcal{B}(\hi))$, and the density 
	of $i_\rho (\mathcal{B}(\hi))$ implies these hold on all of $\mathcal{S}_2(\hi)$.
	
	For the continuity statement, it suffices to assume that $(T_t)_{t \geq 0}$ is 
	weak$^*$-continuous and show $(\widetilde{T_t})_{t \geq 0}$ is strongly continuous on 
	$i_\rho (\mathcal{B}(\hi))$, since 
	$i_\rho (\mathcal{B}(\hi))$ is dense in $\mathcal{S}_2(\hi)$ and $\widetilde{T_t}$ is a
	contraction on $\mathcal{S}_2(\hi)$ 
	for all $t \geq 0$.  To this end, let $x \in \mathcal{B}(\hi)$.  Then
	\begin{align*}
	\|\widetilde{T_t}(i_\rho (x))-i_\rho (x ) \|_2^2 
	& = \| \rho^{1/4}T_t(x) \rho^{1/4} - \rho^{1/4} x  \rho^{1/4} \|_2^2\\
	& = \| \rho^{1/4}T_t(x ) \rho^{1/4} \|_2^2 + \| \rho^{1/4} x \rho^{1/4} \|_2^2 \\ 
	& \quad - \la \rho^{1/4}x \rho^{1/4}, \rho^{1/4}T_t(x )\rho^{1/4} \ra_{\si_2(\hi )} \\
	& \quad -
	\la \rho^{1/4}T_t(x )\rho^{1/4}, \rho^{1/4} x \rho^{1/4} \ra_{\si_2(\hi )}\\ 
	& = \|\widetilde{T_t}(i_\rho (x))\|_2^2 + \|i_\rho (x)\|_2^2 \\
	& \quad - 
	2 \Re \la \rho^{1/4}x \rho^{1/4}, \rho^{1/4}T_t(x ) \rho^{1/4} \ra_{\si_2(\hi )} \\ 
	& \leq 2\|i_\rho (x )\|_2^2 - 
	2 \Re \left( tr(\rho^{1/4} x^* \rho^{1/4} \rho^{1/4}T_t(x)\rho^{1/4} ) \right) \\ 
	& = 2 \Re \left( tr( \rho^{1/4} x^* \rho^{1/4}\rho^{1/4}x  \rho^{1/4}{-}
	\rho^{1/4} x^*\rho^{1/4} \rho^{1/4}T_t(x)\rho^{1/4} ) \!\right) \\ 
	& =2 \Re \left(tr\left( \rho^{1/2} x^*\rho^{1/2} (x  -T_t(x)) \right) \right)
	\rightarrow 0
	\end{align*}
	since $\rho^{1/2} x^*\rho^{1/2}$ is trace-class.  Therefore 
	$(\widetilde{T_t})_{t \geq 0}$ is a strongly continuous 
	semigroup of contractions on $\mathcal{S}_2(\hi)$.
	
	To prove the final statement, first assume that $x \in D(L)$.  Then
	\begin{equation}\label{equationrmk1}
	\textrm{weak}^*-\lim_{t \rightarrow 0} \frac{T_t(x)-x}{t}=L(x).
	\end{equation}
	Notice that for every $y \in \mathcal{S}_2(\hi)$ we obtain, by 
	Proposition~\ref{proposition2}(c), that \linebreak
	$\rho^{1/4}y^* \rho^{1/4} \in \mathcal{S}_1(\hi)$ and therefore the map 
	$\mathcal{B}(\hi) \ni z \mapsto \Tr(z \rho^{1/4}y^* \rho^{1/4}) \in \C$ is 
	weak$^*$ continuous.  Thus Equation~(\ref{equationrmk1}) implies
	\[
	\Tr \left( \rho^{1/4}y^* \rho^{1/4} \frac{T_t(x)-x}{t}  \right) \xrightarrow{t\rightarrow 0} 
	\Tr \left( \rho^{1/4}y^* \rho^{1/4} L(x)  \right);
	\]
	that is,
	\[
	\left\la y , \rho^{1/4} \frac{T_t(x)-x}{t} \rho^{1/4} \right\ra_{\si_2(\hi )} \xrightarrow{t \rightarrow 0} 
	\left\la y , \rho^{1/4} L(x) \rho^{1/4} \right\ra_{\si_2(\hi )},
	\]
	and hence,
	\begin{equation}\label{equationrmk2}
	\left\la 
	y , \frac{ \widetilde{T_t}( \rho^{1/4}x \rho^{1/4})- \rho^{1/4} x \rho^{1/4}}{t} 
	\right\ra_{\si_2(\hi )} 
	\xrightarrow{t \rightarrow 0} \left\la y, \rho^{1/4} L(x) \rho^{1/4} \right\ra_{\si_2(\hi )}.
	\end{equation}
	By \cite[Proposition 3.36]{attal}, we obtain that $\rho^{1/4}x\rho^{1/4} \in D(\widetilde{L})$
	and
	$ \widetilde{L}(\rho^{1/4}x\rho^{1/4})=\rho^{1/4}L(x)\rho^{1/4}.$
	
	Conversely, by the Spectral Theorem there exists an orthonormal
	basis $(h_n)_{n \in \N}$ of $\hi$ formed by eigenvectors of 
	$\rho$. Let $L_{(h_n)}$ denote the generator of $(T_t)_{t \geq 0}$ 
	with respect to $(h_n)_{n \in \N}$. 
	Let $x \in \bi (\hi )$ and assume that 
	$\rho^{1/4}x\rho^{1/4} \in D(\widetilde{L})$. Then 
	we have that 
	\[
	\frac{ \widetilde{T_t}( \rho^{1/4}x \rho^{1/4})- \rho^{1/4} x \rho^{1/4}}{t} 
	\xrightarrow{t \rightarrow 0} \widetilde{L}(\rho^{1/4} x \rho^{1/4}) \quad 
	\textrm{in }\si_2 (\hi ),
	\]
	and hence
	\begin{equation} \label{E:s2conv}
	\rho^{1/4} \frac{T_t(x)-x}{t} \rho^{1/4} 
	\xrightarrow{t \rightarrow 0} \widetilde{L}(\rho^{1/4} x \rho^{1/4}) \quad 
	\textrm{in }\si_2 (\hi ).
	\end{equation}
	We will prove that $x \in D (L_{(h_n)})$. Indeed, we have that \[\langle h,\rho^{1/4} \frac{T_t(x)-x}{t} \rho^{1/4} h'\rangle
	\xrightarrow{t \rightarrow 0} \langle h,\widetilde{L}(\rho^{1/4} x \rho^{1/4})h'\rangle\] for all $h,h'\in\hi$, so for any $n,m\in\N$ we may set $h=(\rho^{1/4})^{(-1)}h_n$ and $h'=(\rho^{1/4})^{(-1)}h_m$ to obtain \begin{align*} \langle (\rho^{1/4})^{(-1)}h_n,\rho^{1/4} \frac{T_t(x)-x}{t} \rho^{1/4}(\rho^{1/4})^{(-1)}h_m\rangle
	\\  \xrightarrow{t \rightarrow 0} \langle(\rho^{1/4})^{(-1)}h_n,\widetilde{L}(\rho^{1/4} x \rho^{1/4})(\rho^{1/4})^{(-1)}h_m\rangle.\end{align*} Noting that $(\rho^{1/4})^\ast=\rho^{1/4},$ $((\rho^{1/4})^{(-1)})^\ast=(\rho^{1/4})^{(-1)}$, and $\rho^{1/4}(\rho^{1/4})^{(-1)}h_k=h_k$ for all $k\in\N$, this implies \[\langle h_n, \frac{T_t(x)-x}{t} h_m\rangle\xrightarrow{t \rightarrow 0} \langle h_n,(\rho^{1/4})^{(-1)}\widetilde{L}(\rho^{1/4} x \rho^{1/4})(\rho^{1/4})^{(-1)}h_m\rangle.\] Because this limit exists for all $n,m\in \N$ we have $x \in D (L_{(h_n)})$, and moreover \[L_{(h_n)}(x)= (\rho^{1/4})^{(-1)}\widetilde{L}(\rho^{1/4} x \rho^{1/4})(\rho^{1/4})^{(-1)}.\]
\qed\end{proof}

\begin{remark} \label{R:morekfgv}
	Since the proof of Equation~(\ref{E:kfgv}) is not included
	in \cite{kfgv}, we want to mention that its proof follows from our Remark~\ref{R:kfgv}
	in a similar way that our Theorem~\ref{theorem4} followed from our 
	Theorem~\ref{proposition3} (even the proof of the strong continuity of the semigroup
	$(\overline{T_t})_{t \geq 0}$ follows the exact same argument as the proof of the strong 
	continuity of the semigroup $(\widetilde{T_t})_{t \geq 0}$ that appeared in 
	Theorem~\ref{theorem4}).
	Moreover, the assumptions that the faithful state is normal and invariant for the semigroup
	and that the operators of the semigroup are completely positive
	that are mentioned in 
	\cite{kfgv} for Equation~(\ref{E:kfgv})  are not needed for its proof,
	because such assumptions were not used in 
	Remark~\ref{R:kfgv}. Instead, for the validity of Equation~(\ref{E:kfgv}), one merely needs to assume that 
	the faithful state is subinvariant for the semigroup of Schwarz maps. 
	Note also that, unlike Equation~(\ref{E:kfgv}),
	Theorem~\ref{theorem4} relates the generators of the two semigroups.
\end{remark}

\section{Applications to Quantum Markov Semigroups and Their Generators}\label{QMS}

Since quantum Markov semigroups (QMSs) are semigroups of completely positive maps
on von~Neumann algebras
(and hence $2$-positive maps and hence Schwarz maps), we naturally obtain applications 
of Theorem~\ref{theorem4} in the study of QMSs.  
The existence of invariant normal states for QMSs has been discussed in \cite{fr1} and
\cite{fr2}. Sufficient conditions for a semigroup to be    
decomposable into a sequence of irreducible semigroups each of them having an invariant normal state are 
given in \cite{Uman06} (see top half of page 608, Theorem~5 on page 608, and
Proposition~5 on page 609).  
There are many results in the literature of semigroups 
which depend on the existence of invariant faithful normal states (for example, see
\cite{frigerio1}, \cite{frigerio2}, \cite{fr3}, 
\cite{fr4}, and \cite{csu2}) and this assumption is often taken for granted as being physically
reasonable. 
QMSs have been extensively studied since the 1970s with the exact form for the
generators being one of the topics 
which has garnered a fair amount of attention. See for example \cite{L}, \cite{GKS},
\cite{ce}, 
\cite{davies2}, \cite{holevo},  \cite{ab}, \cite{az}, and \cite{hsw}.  
The generator of a QMS is a generally unbounded operator defined on a weak$^*$ dense
linear subspace of
$\bi (\hi )$. If the generator is bounded then the semigroup is uniformly continuous and the exact form of the generator 
was found in \cite{GKS} and \cite{L}.
In this section, given a Hilbert space $\hi$ and a QMS on $\bi (\hi)$ having an 
invariant faithful normal state we study the associated semigroup of contractions on 
$\si_2 (\hi )$. In particular, in Theorem~\ref{Thm:final2} we describe the extended generator (and hence generator) of 
such a QMS under the assumption that the generator of the associated semigroup on $\si_2(\hi)$ has compact resolvent.

\begin{definition}
	A \textbf{quantum Markov semigroup (QMS)} on 
	$\bi (\hi)$,
	(for some  Hilbert space $\hi$), is 
	a weak$^*$-continuous one-parameter semigroup of bounded linear 
	operators acting on $\bi (\hi)$,  such that each member of the 
	semigroup is 
	completely positive and identity preserving. 
\end{definition}

\begin{remark}
	If $\hi$ is a Hilbert space and $(T_t)_{t \geq 0}$ is a QMS on 
	$\bi (\hi )$ which has a subinvariant  normal state 
	$\omega_\rho$ for some  $\rho \in \si_1 (\hi )$, then $\omega_\rho$
	is in fact an invariant state for $(T_t)_{t \geq 0}$. Indeed for every
	$t \geq 0$,
	\[
	\Tr (T_t^\dagger (\rho))= \Tr (T_t^\dagger (\rho) 1 )= \Tr (\rho T_t(1))=
	\Tr(\rho 1)= \Tr (\rho ),
	\]
	which together with $T_t^\dagger (\rho) \leq \rho$ implies that 
	$T_t^\dagger (\rho) = \rho$.
\end{remark}

Usually the notion of complete positivity applies to maps on $C^*$-algebras. In particular,
if the $C^*$-algebra
is equal to $\bi (\hi )$ for some Hilbert space $\hi$,  then the notion of 
complete positivity becomes equivalent  
to the following: A map $ \ti :\bi (\hi ) \to \bi (\hi )$ is completely positive if and only if for
every $n \in \N$, 
$x_1, \ldots , x_n \in \bi (\hi )$ and $h_1, \ldots , h_n \in \hi$,
\begin{equation} \label{E:cp}
\sum_{i,j =1}^n \la h_i, \ti (x_i^* x_j) h_j \ra \geq 0 .
\end{equation}
Note that Equation~(\ref{E:cp}) makes perfect sense even if the map $\ti$ is not defined 
on a $C^*$-algebra,
as long as $\ti$ is defined on a Banach $^*$-algebra $\si$ of operators on a
Hilbert space $\hi$. For example, $\si$ can be equal to $\mathcal{S}_2(\hi)$ and 
$\ti$ can be a bounded linear operator from $\si$ to $\si$. We make this
extension of the notion of complete positivity in the next definition.

\begin{definition}
	Let $\hi$ be a Hilbert space and $\si$ be a Banach $^*$-algebra of bounded linear
	operators on $\hi$.  A bounded linear map $\ti: \si \rightarrow \si$ will be called 
	\textbf{completely positive} if for every $n \in \N$, $x_1, \dots ,x_n \in \si$ and 
	$h_1, \dots ,h_n \in \hi$, Equation \eqref{E:cp} holds.
\end{definition}

This terminology will be used in the next result.

\begin{proposition}\label{proposition4.3}		
	Let $(T_t)_{t \geq 0}$ be a weak$^*$-continuous semigroup of Schwarz maps on $\bi (\hi )$ for some Hilbert space $\hi$ which possesses an invariant faithful normal state $\omega_{\rho}$ for some 
	$\rho \in \mathcal{S}_1(\hi)$. Then the operators $T_t$ are completely positive for all $t \geq 0$ if and only if the operators $\widetilde{T_t}$ constructed in Theorem~\ref{theorem4} are completely positive for all $t \geq 0$.
\end{proposition}

\begin{proof} First, assume $T_t$ is completely positive for $t\geq0$, and let $x_1,x_2, \dots , x_n \in \mathcal{B}(\hi)$ and $h_1,h_2, \dots ,h_n \in \hi$. 
	Then	
	\begin{align*}
	& \sum_{i,j=1}^n 
	\left\la h_i, \widetilde{T}_t \left( (\rho^{1/4}x_i \rho^{1/4})^*(\rho^{1/4}x_j \rho^{1/4}) \right)h_j \right\ra  \\
	=&
	\sum_{i,j=1}^n 
	\left\la \rho^{1/4}h_i , T_t \left( (\rho^{1/4}x_i)^*(\rho^{1/4}x_j) \right) \rho^{1/4}h_j \right\ra
	\geq 0
	\end{align*}
	since $T_t$ is completely positive.  Further, since the map $i_\rho$ from 
	Proposition~\ref{proposition2} has dense range, $\widetilde{T}_t$ is completely positive 
	on $\mathcal{S}_2(\hi)$.
	
	Conversely, assume $\widetilde{T}_t$ is completely positive for $t\geq0$, and let $t \geq 0$, $x_1,x_2, \dots , x_n \in \mathcal{B}(\hi)$ and $h_1,h_2, \dots ,h_n \in \hi$. Then \[
	\sum_{i,j=1}^n\left\la \rho^{1/4}h_i,T_t(x_i^*x_j)\rho^{1/4}h_j\right\ra = \sum_{i,j=1}^n\left\la h_j,\widetilde{T}_t((x_i\rho^{1/4})^*(x_j\rho^{1/4}))h_j\right\ra\geq0\]
	since $\widetilde{T}_t$ is completely positive. Because the map $\rho^{1/4}$ has dense range, this is sufficient to show $T_t$ is completely positive.
\qed\end{proof}

For the next result, recall the notion of conditionally completely positive maps introduced by 
Lindblad in \cite{L}.
A linear operator $L : D(L) (\subseteq {\mathcal B}(\hi)) \to {\mathcal B}(\hi)$ is called 
\textbf{conditionally completely positive}
if for all $n \in \N$, for all $a_1,a_2, \dots, a_n \in {\mathcal B}(\hi)$ such that 
$a_i^*a_j \in D(L)$ for all $i,j=1,2,\dots ,n$,
that for all $h_1,h_2, \dots h_n \in \hi$ with $\sum_{i=1}^n a_i(h_i)=0$, we have that
\[
\sum_{i,j=1}^n \la h_i, L(a_i^*a_j)h_j \ra \geq 0 .
\]
The next result is known for uniformly continuous semigroups. For example, see 
\cite[Proposition 3.12 and Lemma 3.13]{fagnola1}, or see \cite[Proposition~2.9]{Evans2}.  
In fact the known proof works for a more general setting as the next result indicates.

\begin{theorem}\label{theoremccp}
	Let $\si$ be a Banach $^*$-algebra of operators acting on a Hilbert space 
	$\hi$.  
	\begin{enumerate}
		\item Let $(\ti_t)_{t \geq 0}$ be a WOT continuous semigroup on $\si$ and let 
		$\li$ be its generator.  
		If $\ti_t$ is completely positive for all $t \geq 0$ then $\li (a^*)=\li (a)^*$ for all 
		$a \in D(\li )$ and $\li$ is conditionally completely positive. 
		
		\item Let $(\ti_t)_{t \geq 0}$ be a uniformly continuous semigroup on $\si$ 
		with generator $\li$.  
		If $\li (a^*)=\li (a)^*$ for all $a \in \si$ and $\li$ is conditionally completely positive,
		then $\ti_t$ is completely positive for all $t \geq 0$.
	\end{enumerate}
\end{theorem}

\begin{proof}
	The proof of (2) immediately follows from 
	\cite[Proposition 3.12 and Lemma 3.13]{fagnola1}.  To prove (1),  suppose 
	$a_1,a_2, \dots , a_n \in \si$ such that $a_i^*a_j \in D(L)$ for all $i,j=1,\dots , n$ 
	and $h_1,h_2, \dots ,h_n \in \hi$ such that 
	$\sum_{i=1}^n a_i (h_i) =0$.  Then,
	\begin{align*}
	\sum_{i,j=1}^n 
	\la h_i, \li (a_i^*a_j)h_j \ra & = 
	\lim_{t \rightarrow 0^+} \sum_{i,j=1}^n \frac{1}{t} \la h_i, (T_t-1)(a_i^*a_j)h_j \ra \\ 
	& =\lim_{t \rightarrow 0^+} \sum_{i,j=1}^n \frac{1}{t} \la h_i, T_t(a_i^*a_j)h_j \ra 
	\quad   (\textrm{since } \sum_{i=1}^na_i(h_i)=0) \\ 
	& \geq 0  
	\end{align*}
	since $T_t$ is completely positive for all $t \geq 0$.
\qed\end{proof}

\begin{corollary}\label{corollary2}		
	Let $\hi$ be a Hilbert space and  $(T_t)_{t \geq 0}$ be a QMS on $\mathcal{B}(\hi)$ which possesses
	an invariant faithful normal state $\omega_{\rho}$ for some $\rho \in \mathcal{S}_1(\hi)$.  
	Let $\widetilde{L}$ be the generator of the strongly continuous semigroup $(\widetilde{T}_t)_{t \geq 0}$ of contractions on $\mathcal{S}_2(\hi)$ defined in Theorem~\ref{theorem4}.  Then $\widetilde{L}(a^*)= \widetilde{L}(a)^*$
	for all $a \in D(\widetilde{L})$ and $\widetilde{L}$ is conditionally completely positive.
\end{corollary}

\begin{proof}
	The proof follows immediately from Proposition~\ref{proposition4.3} and 
	Theorem~\ref{theoremccp}(1).
\qed\end{proof}

\subsection{The Form of $\boldsymbol{L_{(h_n)}}$ when the Resolvent of $\boldsymbol{\widetilde{L}}$ is Compact}\label{compactresolvent}

In this subsection we consider the form of the extended generator $L_{(h_n)}$ when the resolvent of $\widetilde{L}$ is compact, by which we mean that $(\widetilde{L}-\lambda)^{-1}$ is compact for some $\lambda$ in the resolvent set of $\widetilde{L}$ (equivalently all $\lambda$ in the resolvent set, by the resolvent identity). Notably, this assumption implies $\widetilde{L}$ is necessarily unbounded if $\hi$ is infinite dimensional, since the composition of bounded with compact is compact but $I=(\widetilde{L}-\lambda)(\widetilde{L}-\lambda)^{-1}$ is not compact. More information about operators on a Hilbert space with compact resolvent can be found in Theorems XIII.4.1 and XIII.4.2 of \cite{DS}. Examples of such operators are the diagonal operator with eigenvalues $1,2,3,\ldots$, or the Laplacian on a bounded domain of $\R^d$.

We will make use of the following two notations:

\noindent \textbf{Notation:}
	Let $\hi$ be a Hilbert space and  $w,z \in \mathcal{S}_2(\hi)$.  Define 
	$M_w \otimes z :\mathcal{S}_2(\hi) \otimes \hi \rightarrow \mathcal{S}_2(\hi) \otimes \hi$ 
	by
	\[
	M_w \otimes z \left( \sum_{i=1}^k x_i \otimes h_i \right)=\sum_{i=1}^kx_iw \otimes z(h_i) .
	\]

\noindent \textbf{Notation:}
	Let $\hi$ be a Hilbert space and $e \in \hi $ such that $\| e \|=1$. 
	Define $T_e :\mathcal{S}_2(\hi) \otimes \hi \rightarrow \mathcal{S}_2(\hi) \otimes \hi$ by
	\[
	T_e  \left(\sum_{i=1}^k x_i \otimes h_i \right)= \sum_{i=1}^k |x_i(h_i) \ra \la e| \otimes e .
	\]

We are now ready to state the main result:

\begin{theorem} \label{Thm:final2}
	Let $\hi$ be a Hilbert space, $(T_t)_{t \geq 0}$ be a QMS on 
	$\mathcal{B}(\hi)$ having an invariant faithful normal state 
	$\omega_{\rho}$ for some 
	$\rho \in \mathcal{S}_1(\hi)$, and $L$ be the generator of 
	$(T_t)_{t \geq 0}$.  
	Let $(\widetilde{T}_t)_{t \geq 0}$ be the strongly continuous
	semigroup of contractions on $\mathcal{S}_2(\hi)$ defined in Theorem~\ref{theorem4} and
	let $\widetilde{L}$ be its generator. Assume that the generator $\widetilde{L}$ has compact resolvent.
	Then the following assertions are valid:
	\begin{itemize}		
		\item[(a)] There exist complete orthonormal families 
		$(a_n)_{n \in \N}$ and $(b_n)_{n \in \N}$ of self-adjoint elements in $\si_2(\hi )$ and a sequence of positive scalars $(\lambda_n)_{n\in\N}$ with $\lambda_{n}\to\infty$ as $n\to\infty$ (if $\hi$ is infinite dimensional)
		such that 
		\begin{equation}\label{genform2}
		\widetilde{L}= I+\sum_{n=1}^\infty  \lambda_{n}|a_n \ra \la b_n |
		\end{equation}
		where the sum converges in the SOT (if it is infinite), i.e. for every $x \in D(\widetilde{L})$ we have that $\widetilde{L}(x)= x + 
		\sum_{n=1}^\infty \lambda_n \la b_n , x \ra a_n$ with $\sum_n |\lambda_n \la b_n, x \ra |^2 < \infty$.
		
		\item[(b)] By the Spectral Theorem there is an orthonormal
		basis $(h_n)_{n \in \N}$ of $\hi$ which consists of eigenvectors 
		of $\rho$. Let $L_{(h_n)}$ denote the generator of 
		$(T_t)_{t \geq 0}$ with respect to $(h_n)_{n \in \N}$. Then 
		\begin{equation} \label{compactresolventL}
		L_{(h_n)}
		=I+\sum_{n=1}^\infty \lambda_n|i_{\rho^{(-1)}}(a_n) \ra \la i_\rho(b_n) |
		\end{equation} where the sum converges in the SOT (if it is infinite). We note that \linebreak $|i_{\rho^{(-1)}}(a_n) \ra\la i_\rho(b_n)|$ has domain $\bi(\hi)$ for all $n$, since \[|i_{\rho^{(-1)}}(a_n) \ra\la i_\rho(b_n)|x=\la i_\rho(b_n),x\ra i_{\rho^{(-1)}}(a_n)=\la b_n,i_\rho(x)\ra i_{\rho^{(-1)}}(a_n)\] and $b_n,i_\rho(x)\in \si_2(\hi)$. Explicitly, for every $x \in D(L_{(h_n)})$ and every $i,j\in \N$ we have that \[\la h_i,[L_{(h_n)}(x)]h_j\ra = \la h_i,xh_j\ra+\sum_{n=1}^\infty \lambda_n \la b_n,i_\rho(x)\ra \la h_i, i_{\rho^{(-1)}}(a_n) h_j\ra. \]
		
		\item[(c)] We have $$I=-\sum_{n=1}^\infty \lambda_n\la b_n,\rho^{1/2}\ra i_{\rho^{(-1)}}(a_n),$$ where the sum converges in the SOT (if it is infinite).	
	
		\item[(d)] For all $e \in \hi$ with 
		$\| e \| =1$ we have that the operator 
		$\widetilde{L}_{\otimes,e} : \si_2 (\hi )\otimes \hi  \to \si_2 (\hi )\otimes \hi $ 
		is positive, where the operator $\widetilde{L}_{\otimes,e}$ is defined by 
		\begin{equation} \label{E:DLe2}
		\widetilde{L}_{\otimes,e}= (Id +T_e^*)
		\left( \sum_{n=1}^\infty \lambda_n M_{b_n} \otimes a_n\right) 
		(Id+T_e)
		\end{equation}
		where $Id$ stands for the identity operator on $\si_2 (\hi ) \otimes \hi$ and the sum converges in the SOT (if it is infinite).
	\end{itemize}
\end{theorem}

We note that the sum in \eqref{genform2} is finite if and only if $\hi$ is finite dimensional. Indeed, if $\widetilde{L}$ is bounded with compact resolvent then $\hi$ is finite dimensional, as remarked in the preamble of this subsection. The proof of Theorem~\ref{Thm:final2} is at the end of this subsection, after the following three results:

\begin{lemma} \label{Lem:selfadjoint matrices} Let $\hi$ be a separable Hilbert space and $A$ be an invertible linear operator on $S_2(\hi)$ with dense domain which is closed under adjoints. If $A$ satisfies $A(a^*)=(A(a))^\ast$ for all $a\in~D(A)$, then $D(A^\dag)$ and $D(A^{-1})$ are closed under adjoints, $A^\dag(b^\ast)=(A^\dag(b))^\ast$ for all $b\in D(A^\dag)$, and $A^{-1}(c^\ast)=(A^{-1}(c))^\ast$ for all $c\in D(A^{-1})$.\end{lemma}

\begin{proof}
	Let $a\in D(A)$ and $b\in D(A^\dag)$. Then \begin{align*}|\la A(a),b^*\ra| &=|\la (A(a^*))^*,b^*\ra|=|\la b,A(a^*)\ra| \\ &=|\la A^\dag(b),a^*\ra| =|\la a,(A^\dag(b))^*\ra|\leq ||a|| ||(A^\dag(b))^*||,\end{align*} and so $b^\ast\in D(A^\dag)$ by definition. As before, \[\la a, A^\dag(b^*)\ra=\la A(a),b^*\ra=\la a,(A^\dag(b))^*\ra,\] and since $D(A)$ is dense this implies $A^\dag(b^\ast)=(A^\dag(b))^\ast$ for all $b\in D(A^\dag)$. Further, for every $c\in D(A^{-1})$ there exists an $a\in D(A)$ such that $A(a)=c$. Since $A$ is star-preserving we have that $A(a^*)=c^*$. Then, by definition, $(A^{-1}(c))^*=a^*=A^{-1}(c^*)$. 
\qed\end{proof}

\begin{lemma} \label{Lem:selfadjoint matrices2} Let $\hi$ be a Hilbert space and $A$ be a compact and self-adjoint linear operator on $S_2(\hi)$. Then $A$ satisfies $A(a^*)=(A(a))^\ast$ for all $a\in~S_2(\hi)$ if and only if \[A=\sum_{n=1}^\infty\lambda_{n}|x_n\ra\la x_n|\] with $(\lambda_n)_{n=1}^\infty\subseteq \R$ and $(x_n)_{n=1}^\infty$ an orthonormal basis of $S_2(\hi)$ consisting of self-adjoint operators. \end{lemma}

\begin{proof}
	If $A$ is compact and self-adjoint, then the Spectral Theorem implies there is an eigensystem decomposition \[A=\sum_{n=1}^\infty \lambda_n |y_n\ra\la y_n|,\] with $(\lambda_n)_{n=1}^\infty\subseteq \R$ and $(y_n)_{n=1}^\infty$ an orthonormal basis of $S_2(\hi)$. Because $A$ is self-adjoint and star-preserving, we have that $A(y_n)=\lambda_n y_n$ implies $A(y_n^\ast)=\lambda_n y_n^\ast$. Thus, every eigenspace of $A$ is self-adjoint. Let $E$ be an eigenspace of $A$ corresponding to a fixed eigenvalues, and consider the orthonormal basis $(y_{n_j})_{j=1}^N\subseteq(y_n)_{n=1}^\infty$ of $E$. Because $E$ is self-adjoint, from $\la y_{n_j},y_{n_k}\ra=\la y_{n_j}^*,y_{n_k}^*\ra=\delta_{jk}$ it follows that $(y_{n_j}^*)_{j=1}^N$ is an orthonormal basis of $E$. Define self-adjoint operators $a_j=y_{n_j}+y_{n_j}^*$ and $a_{N+j}=i(y_{n_j}-y_{n_j}^*)$ for each $1\leq j\leq N$ so that $E=\Span(a_j)_{j=1}^{2N}$. We follow the Gram-Schmidt process and set $b_1=a_1$ and recursively define \[b_k=a_k-\sum_{j=1}^{k-1}\frac{\la b_j,a_k\ra}{\la b_j,b_j\ra}b_j\] to produce a sequence of $N$ many orthogonal operators which span $E$ (the remaining $N$ many operators produced by the Gram-Schmidt process become zero). Straight forward calculation reveals that $\la a_j,a_k\ra$ is real for every $1\leq j,k\leq 2N$, and hence $\la b_j,a_k\ra$ is real for every $1\leq j,k\leq 2N$. Each $b_k$ is thus self-adjoint as a real combination of self-adjoint operators, and the set $(b_k)_{k=1}^N$ can therefore be normalized to a set of self-adjoint orthonormal operators $(x_j)_{j=1}^N$ which span $E$. Replacing $y_n$ with $x_n$ in the original eigensystem decomposition for each eigenspace $E$, we have \[A=\sum_{n=1}^\infty \lambda_n |x_n\ra\la x_n|,\] as desired.
\qed\end{proof}

\begin{lemma} \label{Lem:Mpositive}
	Let $\hi$ be a Hilbert space and $\widetilde{L}$ be a bounded linear operator 
	on $\si_2 (\hi )$ which has the form \eqref{genform2}. Then $\widetilde{L}$ is conditionally completely positive if and only if
	for some (equivalently all) normalized vector $e\in \hi$, 
	the operator $\widetilde{L}_{\otimes,e}$ defined on the Hilbert space $\si_2(\hi ) \otimes \hi$,
	by Equation~(\ref{E:DLe2}), is positive.
\end{lemma}

\begin{proof} First note that $I+A$ is conditionally completely positive if and only if $A$ is (as is easily verified), so for simplicity we may assume instead that $\widetilde{L}= \sum_{n=1}^\infty  \lambda_{n}|a_n \ra \la b_n |$.

	We will start with the forward direction and suppose $\widetilde{L}$ is conditionally completely positive.  Let $e \in \hi$ with $\| e \| =1$.  Since  
	$W= \{ \sum_{i=1}^k y_i \otimes h_i' : y_i \in \mathcal{S}_2(\hi), 
	h_i' \in \hi \}$ is dense in
	$\mathcal{S}_2(\hi) \otimes \hi$, in order to verify that 
	$\widetilde{L}_{\otimes,e} \geq 0$ it is enough to consider an element 
	$w=\sum_{i=1}^k y_i \otimes h_i' \in W$ and verify that 
	$\la w, \widetilde{L}_{\otimes,e} w \ra_\otimes \geq 0$, where $\la \cdot , \cdot \ra_\otimes$ 
	will denote the inner product of $\si_2(\hi) \otimes \hi$.  (The 
	reason that we chose $h_i'$ to denote a generic element of $\hi$
	is because we have used $h_n$ to denote the orthonormal 
	eigenvectors of $\rho$ in the statement of Theorem~\ref{Thm:final2}).
	We will denote the inner product
	of $\hi$ by $\la \cdot , \cdot \ra_\hi$. Fix 
	$w=\sum_{i=1}^k y_i \otimes h_i' \in W$ and let 
	$v=-\sum_{i=1}^k y_i(h_i')$.  Define 
	$y_{k+1}=|v \ra \la e|$ and $h_{k+1}'=e$.  Then 
	$\sum_{i=1}^{k+1}y_i(h_i')=0$ and, since 
	$\widetilde{L}$ is conditionally completely positive, we have that
	\begin{align*}
	0  \leq & \sum_{i,j=1}^{k+1} \la h_i' , \widetilde{L}(y_i^*y_j)h_j' \ra_\hi \\  
	= & \sum_{i,j=1}^{k+1} 
	\sum_{n=1}^\infty \lambda_n\Tr(y_i^*y_jb_n) \la h_i' ,a_n(h_j') \ra_\hi\\ 
	= & \sum_{i,j=1}^{k+1} 
	\sum_{n=1}^\infty\lambda_n\la y_i \otimes h_i' , y_jb_n \otimes a_n(h_j') \ra_\otimes \\ 
	= & \sum_{i,j=1}^{k+1} \sum_{n=1}^\infty 
	\left\la 
	y_i \otimes h_i' , \lambda_nM_{b_n} \otimes a_n (y_j \otimes h_j') 
	\right\ra_\otimes
	\\ 
	= & \left\la 
	\sum_{i=1}^{k+1} y_i \otimes h_i' , 
	\left(\sum_{n=1}^\infty \lambda_nM_{b_n} \otimes a_n \right)
	\left( \sum_{j=1}^{k+1} y_j \otimes h_j' \right) 
	\right\ra_\otimes .
	\end{align*}
	Notice that
	\begin{align*}
	\sum_{i=1}^{k+1}y_i \otimes h_i' & = 
	\sum_{i=1}^k y_i \otimes h_i' + y_{k+1} \otimes h_{k+1}'=
	w  - \sum_{i=1}^k |y_i(h_i') \ra \la e| \otimes e \\
	&=w -T_e \left( \sum_{i=1}^k y_i \otimes h_i' \right) =(Id -T_e) (w)
	\end{align*}
	where $Id$ denotes the identity operator on $\si_2 (\hi ) \otimes \hi$,
	which finishes the proof of the forward direction.
	
	For the other direction, suppose that $\widetilde{L}_{\otimes,e} \geq 0$ for some $e \in \hi$
	with $\| e \| =1$.  Let $k \in \N$, $y_1, \dots , y_k \in \mathcal{S}_2(\hi)$ 
	and $h_1', \dots , h_k' \in \hi$ such that $\sum_{i=1}^ky_i(h_i')=0$.
	Let $w=\sum_{i=1}^k y_i \otimes h_i' \in \si_2 (\hi )\otimes \hi$. Then,
	\begin{align} \label{DLegeq0}
	0 &\leq  \la w , \widetilde{L}_{\otimes,e} (w) \ra_\otimes 
	\\ &= \left\la w,  
	(Id - T_e)^*
	\left( \sum_{n \in N}\lambda_nM_{b_n} \otimes a_n \right)
	(Id -T_e) (w) 
	\right\ra_\otimes \nonumber \\ 
	&= \left\la 
	(Id -T_e) w, 
	\left(\sum_{n \in N}\lambda_n M_{b_n} \otimes a_n\right)
	(Id -T_e) (w) 
	\right\ra_\otimes .
	\end{align}
	Notice that 
	\[
	T_e (w) =T_e \left( \sum_{i=1}^k y_i \otimes h_i' \right)=
	\left| \sum_{i=1}^k y_i(h_i') \right\ra \left\la e \right| \otimes e = |0 \ra \la e | \otimes e = 0.
	\]
	Hence Inequality~(\ref{DLegeq0}) gives
	\begin{align*}
	0 &\leq   
	\left\la 
	Id  (w), 
	\left(\sum_{n \in N} \lambda_nM_{b_n} \otimes a_n \right)  
	Id  (w) 
	\right\ra_\otimes \\
	& = 
	\sum_{i,j=1}^{k} \sum_{n\in N}
	\left\la 
	y_i \otimes h_i' ,\lambda_n M_{b_n} \otimes a_n (y_j \otimes h_j') 
	\right\ra_\otimes \\ 
	& =  \sum_{i,j=1}^{k} \sum_{n\in N} \lambda_n
	\la y_i \otimes h_i' , y_jb_n \otimes a_n(h_j') \ra_\otimes\\ 
	& =   \sum_{i,j=1}^{k} \sum_{n\in N}\lambda_n
	\Tr(y_i^*y_jb_n) \la h_i' ,a_n(h_j') \ra_\hi \\ 
	& =  \sum_{i,j=1}^{k+1} \la h_i , \widetilde{L}(y_i^*y_j)h_j' \ra_\hi .
	\end{align*}
	Therefore $\widetilde{L}$ is conditionally completely positive.  This completes the proof.\qed\end{proof}

The proof of Lemma~\ref{Lem:Mpositive} reveals the following: \begin{remark} Let $\ai=\{\sum_{i=1}^k y_i\otimes h_i'\in\bi(\hi)\otimes \hi:\sum_{i=1}^ky_i(h_i')=0\}$. Then
	\begin{itemize}
		\item For every $w=\sum_{i=1}^ky_i\otimes h_i'\in \bi(\hi)\otimes \hi$ there exists $y_{k+1}\in \bi(\hi)$ and $h_{k+1}'\in \hi$ such that $\sum_{i=1}^{k+1}y_i\otimes h_i'\in\ai$ and $(Id-T_{h_{k+1}'})(w)=\sum_{i=1}^{k+1}y_i\otimes h_i'$.
		\item If a bounded operator $\widetilde{L}$ on $\hi$ has form \eqref{genform2} then $\widetilde{L}$ is completely positive if and only if the operator $\sum_{n=1}^\infty \lambda_nM_{b_n}\otimes a_n:S_2(\hi)\otimes \hi \to S_2(\hi)\otimes \hi$ is positive.
		\item For every $e\in \hi$ we have $\ai\subseteq \ker T_e$.
	\end{itemize}
\end{remark}

We are now ready to present the

\begin{proof}[Proof of Theorem~\ref{Thm:final2}] Since $\widetilde{L}$ generates a strongly continuous semigroup of contractions, we have that $\lambda\in\rho(\widetilde{L})$ for all $\lambda>0$ by the Hille-Yosida Generation Theorem (e.g. Theorem~3.5 of \cite{EN}). Further, $D(\widetilde{L})$ is dense in $S_2(\hi)$ by Theorem~3.1.16~of~\cite{br} and $\widetilde{L}$ is star-preserving by Corollary~\ref{corollary2}, and so $K:=(\widetilde{L}-I)^{-1}$ is star-preserving by Lemma~\ref{Lem:selfadjoint matrices} as the inverse of a star-preserving map with dense domain. Because $\widetilde{L}$ has compact resolvent by assumption, we have that $K$ is furthermore compact. Thus, $K^\dag K$ is compact, self-adjoint, and star-preserving, and so Lemma~\ref{Lem:selfadjoint matrices2} implies \[K^\dag K=\sum_{n=1}^\infty \sigma_n^2|v_n\ra\la v_n|,\] where $\{\sigma_n^2\}_{n\in \N}$ are the nonzero eigenvalues of $K^\dag K$ corresponding to the system $\{v_n\}_{n\in\N}$ of self-adjoint orthonormal eigenoperators. This notation is chosen so that, following Section~2.2 of \cite{EnglHanke}, the singular value expansion of $K$ can be written \[K=\sum_{n=1}^\infty \sigma_n|u_n\rangle\langle v_n|,\] where $\{u_n\}_{n\in\N}$ are self-adjoint orthonormal eigenoperators of $KK^\dag$ given by the relation $\sigma_nu_n:=Kv_n$. By Theorem~2.8 of \cite{EnglHanke} we have that \[\widetilde{L}-I=K^{(-1)}=\sum_{n=1}^\infty\frac{1}{\sigma_n}|v_n\rangle\langle u_n|,\] and hence \[\widetilde{L}=I+\sum_{n=1}^\infty\frac{1}{\sigma_n}|v_n\rangle\langle u_n|,\] proving \eqref{genform2}. Equation \eqref{compactresolventL} follows from \eqref{genform2} and \eqref{LtotildeL}. Part (c) is valid since $L(I)=0$, hence $L_{(h_n)}(I)=0$, and hence $$I+\sum_{n=1}^\infty \lambda_n\la h_n,\rho^{1/2}\ra i_{\rho^{(-1)}}(a_n)=0.$$ Part (d) follows from Lemma~\ref{Lem:Mpositive}; indeed, $(\widetilde{T}_t)_{t\geq0}$ is a completely positive semigroup by Proposition~\ref{proposition4.3}, and so $\widetilde{L}$ is conditionally completely positive by Theorem~\ref{theoremccp}.\qed\end{proof}

\begin{remark} Let $\hi$ be a finite dimensional Hilbert space, and let $(T_t)_{t \geq 0}$ be a QMS on $\mathcal{B}(\hi)$. Then $(T_t)_{t \geq 0}$ possesses an invariant faithful normal state \cite[Theorem 4.2]{fr5} and so Theorem~\ref{Thm:final2} applies. The sum in \eqref{compactresolventL} is finite, and hence by linearity/conjugate linearity we can write \begin{align*}
	L_{(h_n)} &= I+\sum_{n}\lambda_n|i_{\rho^{(-1)}}(a_n^+-a_n^-)\ra\la i_\rho(b_n^+-b_n^-)| \\ &= I+\sum_{m}\lambda_m'|i_{\rho^{(-1)}}(c_m)\ra\la i_\rho(d_m)|, \end{align*} where $c_m\in\{a_n^+,a_n^-:n\},d_m\in\{b_n^+,b_n^-:n\}$, and $\lambda_m'\in\{\pm\lambda_n:n\}$. In particular, $c_m,d_m\geq0$ and $\lambda_m'\in \R$. Thus, for every $x\in \bi(\hi)$, \begin{align*}
	|i_{\rho^{(-1)}}(c_m)\ra\la i_\rho(d_m)| x &= \la i_\rho(d_m),x\ra i_{\rho^{(-1)}}(c_m) \\&= \la \rho^{1/4}d_m\rho^{1/4},x\ra \rho^{-1/4}c_m\rho^{-1/4} \\
	&= \rho^{-1/4}\sqrt{c_m}\Tr\left(\rho^{1/4}d_m\rho^{1/4}x\right)\sqrt{c_m}\rho^{-1/4}\\
	&= \rho^{-1/4}\sqrt{c_m}\Tr\left(\sqrt{d_m}\rho^{1/4}x\rho^{1/4}\sqrt{d_m}\right)\sqrt{c_m}\rho^{-1/4}.\end{align*}

Now, fix any orthonormal basis $(E_k)_k$ of $\bi(\hi)$. Then for every $A\in \bi(\hi)$ we have $$\Tr(A)=\sum_k E_k AE_k^*.$$ Thus,  $$\Tr\left(\sqrt{d_m}\rho^{1/4}x\rho^{1/4}\sqrt{d_m}\right)=\sum_kE_k\sqrt{d_m}\rho^{1/4}x\rho^{1/4}\sqrt{d_m}E_k^*,$$ and so $$L_{(h_n)}(x)=x+\sum_{m,k}\lambda_m'\rho^{-1/4}\sqrt{c_m}E_k\sqrt{d_m}\rho^{1/4}x\rho^{1/4}\sqrt{d_m}E_k^*\sqrt{c_m}\rho^{-1/4}.$$ By defining $y_\ell=\rho^{-1/4}\sqrt{c_m}E_k\sqrt{d_m}\rho^{1/4}$ and $\lambda_\ell'$ to be the corresponding $\lambda_m'$, we obtain $$L_{(h_n)}(x)=x+\sum_\ell \lambda_\ell' y_\ell x y_\ell^*.$$ Moreover, since $L_{(h_n)}(I)=L(I)=0$ we obtain as in the proof of Theorem~\ref{Thm:final2}(c) that $$I+\sum_\ell \lambda_\ell y_\ell y_\ell^*=0.$$ Thus $$ L_{(h_n)}(x)=\frac{1}{2}Ix+\frac{1}{2}xI+\sum_\ell \lambda_\ell' y_\ell x y_\ell^*= \sum_\ell \lambda_\ell'\left(y_\ell x y_\ell^* -\frac{1}{2}\{y_\ell y_\ell^*,x\}\right)$$ which resembles the standard GKSL form developed in \cite{GKS} and \cite{L}, except with Hamiltonian part zero and without the demand that the jump operators $y_\ell$ are traceless. A comparison of such GKSL form to the standard GKSL form, including conversion between the two, is discussed in section 2.1 of \cite{AW}.\end{remark}

\section{Conflict of interest} The authors declare that they have no conflict of interest.

\begin{acknowledgements} We would like to thank Franco Fagnola. His contributions to the results of this work were vital from its conception to the final touches. Without his help the existence of this work would not be possible. We would also like to thank the referee for pointing our attention to a mistake in the original version of this article.
\end{acknowledgements}

\bibliographystyle{spmpsci}      
\bibliography{references}   

\end{document}